\documentclass[11pt]{article}
\usepackage{fullpage}
\usepackage{amsmath}
\usepackage{amssymb}
\usepackage{algorithmic}
\usepackage{algorithm}
\usepackage{graphicx}
\usepackage{xcolor}
\usepackage{latexsym}
\usepackage{amsmath}
\usepackage{amssymb}
\usepackage{amsthm}
\usepackage{hyperref}
\usepackage{cite}
\usepackage{graphicx}
\usepackage{color}
\usepackage{framed}

\newtheorem{theorem}{Theorem}

\newtheorem{theorem*}{Theorem}
\newtheorem{corollary}[theorem]{Corollary}
\newtheorem{lemma}[theorem]{Lemma}
\newtheorem{observation}[theorem]{Observation}
\newtheorem{proposition}[theorem]{Proposition}
\newtheorem{definition}[theorem]{Definition}
\newtheorem{claim}[theorem]{Claim}
\newtheorem{fact}[theorem]{Fact}

\theoremstyle{definition}
\newtheorem{problem}{Problem}

\newtheorem{remark}[theorem]{Remark}

\renewcommand{\vec}[1]{\mathbf{#1}}

\newcommand{\GL}{\mathrm{GL}}

\newcommand{\F}{\mathbb{F}}

\newcommand{\E}{\mathbb{E}}

\newcommand{\K}{\mathbb{K}}
\newcommand{\Z}{\mathbb{Z}}
\newcommand{\Q}{\mathbb{Q}}
\newcommand{\C}{\mathbb{C}}
\newcommand{\R}{\mathbb{R}}
\newcommand{\N}{\mathbb{N}}
\newcommand{\M}{\mathrm{M}}

\newcommand{\D}{\mathrm{D}}

\newcommand{\tr}{\mathrm{tr}}

\newcommand{\rk}{\mathrm{rk}}
\newcommand{\ncrk}{\mathrm{ncrk}}

\newcommand{\im}{\mathrm{im}}

\newcommand{\poly}{\mathrm{poly}}
\newcommand{\fA}{\mathfrak{A}}

\newcommand{\codeg}{\mathrm{codeg}}

\newcommand{\cA}{\mathcal{A}}
\newcommand{\cB}{\mathcal{B}}
\newcommand{\vB}{\mathbf{B}}
\newcommand{\cC}{\mathcal{C}}

\newcommand{\sB}{\mathsf{B}}
\newcommand{\sA}{\mathsf{A}}
\newcommand{\bA}{\mathbf{A}}

\newcommand{\QC}{\mathrm{QC}}

\newcommand{\IQC}{\mathrm{IQC}}

\newcommand{\Adj}{\mathrm{Adj}}
\newcommand{\MISet}{\mathrm{MI}}
\newcommand{\MINum}{\mathrm{NMI}}
\newcommand{\MMINum}{\mathrm{MaxNMI}}

\newcommand{\zerovec}{\mathbf{0}}
\newcommand{\rad}{\mathrm{rad}}

\title{
From independent sets and vertex colorings
to
isotropic spaces and isotropic decompositions \\
[0.2em]
\Large
Another bridge between graphs and alternating matrix spaces
}
\author{
Xiaohui Bei\thanks{School of Physical and Mathematical Sciences, Nanyang
Technological University, Singapore. Email address: {\tt xhbei@ntu.edu.sg}.}
\and
Shiteng Chen\thanks{State Key Laboratory of Computer Science, Institute of
Software, Chinese Academy of Sciences, Beijing,
China. Email address: {\tt chenst@ios.ac.cn}. }
\and
Ji Guan \thanks{State Key Laboratory of Computer Science, Institute of Software,
Chinese Academy of Sciences, Beijing,
China. Email address: {\tt guanji1992@gmail.com}. }
\and
Youming Qiao \thanks{Centre for Quantum Software and Information,
 University of Technology Sydney, Australia. Email address: {\tt
 jimmyqiao86@gmail.com}.}
\and
Xiaoming Sun \thanks{Institute of Computing Technology, Chinese Academy of
Sciences and
University of Chinese Academy of Sciences, Beijing, China. Email address: {\tt
sunxiaoming@ict.ac.cn}.}
}

\date{\today}

\begin{document}

\pagenumbering{gobble}  

\maketitle

\begin{abstract}
In the 1970's, Lov\'asz built a bridge between graphs and alternating
matrix
spaces, in the context of perfect matchings (FCT 1979).
A similar
connection between bipartite graphs and matrix spaces
plays a key role in the recent resolutions of the
non-commutative rank problem (Garg-Gurvits-Oliveira-Wigderson, FOCS 2016; 
Ivanyos-Qiao-Subrahmanyam, ITCS 2017).
In this paper, we lay the foundation for another bridge between graphs and
alternating
matrix spaces, in
the context of
independent sets and vertex colorings. The corresponding structures in
alternating matrix
spaces are \emph{isotropic spaces} and \emph{isotropic decompositions}, both
useful structures in group theory and manifold theory.

We first show that the maximum independent set problem
and the vertex $c$-coloring problem reduce to the maximum isotropic space problem
and the isotropic $c$-decomposition problem, respectively.
Next, we show that several topics and results about independent sets and vertex
colorings have natural correspondences for isotropic spaces and decompositions.
These include algorithmic problems, such as the maximum independent set problem
for bipartite graphs, and exact exponential-time algorithms for the chromatic
number, as well as mathematical questions, such as the number of maximal
independent
sets, and the relation between the
maximum degree and the chromatic number.
These connections lead to
new interactions between graph theory and algebra. Some results have
concrete applications to
group theory and manifold theory, and we initiate a variant of these
structures in the context of quantum information theory. Finally, we propose
several open questions
for further exploration.

\center
{\large \it Dedicated to the memory of Ker-I Ko}
\end{abstract}

\newpage

\tableofcontents

\newpage
\pagenumbering{arabic}     

\section{Introduction}

\subsection{Between graphs and matrix spaces: some known
bridges}\label{subsec:known_connection}

\paragraph{The bridge between perfect matchings and full-rank matrices.}
It is well-known that some
graph-theoretic problems reduce to certain
problems about linear spaces of matrices. A classical example, tracing back to
Tutte \cite{Tut47}, and then more systematically examined by Lov\'asz
\cite{Lov79,Lov89}, concerns perfect matchings. 

Let $\F$ be a field, and $[n]:=\{1, \dots, n\}$. Let
$G=([n],
E)$ be a simple and undirected graph, so $E$ can be viewed as a subset of $\{\{i,
j\}: i, j\in[n], i\neq j\}$. For $n\in\N$ and $\{i, j\}$ where $i, j\in[n]$, $i<
j$, the
elementary
alternating\footnote{An $n\times n$
matrix $A$ over $\F$ is \emph{alternating}, if for
any $v\in \F^n$, $v^tAv=0$. An alternating matrix is always skew-symmetric (i.e. 
$A^t=-A$), and a skew-symmetric matrix is also alternating over fields of 
characteristic not $2$.}
matrix $A_{i,j}$ of size $n\times n$ is the matrix with the $(i, j)$th entry being
$1$, the $(j, i)$th
entry being $-1$, and
the rest entries being $0$. Let $\cA_G$ be the linear space of alternating matrices
spanned by
$A_{i,j}$,
$\{i, j\}\in E$. Then when the field is large enough, $G$
has a perfect matching if and only if $\cA_G$ contains a full-rank matrix.
%

A similar construction for bipartite graphs is also classical.
Let $G=(L\cup R, E)$ be a bipartite graph where $L=R=[n]$, so $E$ can be viewed as 
a subset of $[n]\times [n]$. For $n\in\N$ and $i,
j\in [n]$, the
elementary matrix $E_{i,j}$ of size $n\times n$ is the matrix with
the $(i, j)$th entry being $1$, and the rest entries being $0$. Let $\cB_G$ be the
linear
space of matrices spanned by $E_{i,j}$, $(i, j)\in E$. Then when the field is
large enough, $G$ has a perfect
matching if and only if $\cB_G$ contains a full-rank matrix.

As noted by Lov\'asz \cite{Lov79}, these observations give efficient
randomized
algorithms for deciding the existence of perfect matchings on bipartite graphs and
graphs over a large enough $\F$ via the celebrated Schwartz-Zippel lemma 
\cite{Zip79,Sch80}.
Furthermore, because the determinant polynomial can be
evaluated efficiently in parallel \cite{Ber84,Chi85}, these are actually 
randomized NC algorithms. 

This work of Lov\'asz has inspired several prominent results, including randomized 
NC algorithms for constructing perfect matchings \cite{KUW86,MVV87}, and the 
recent breakthrough of quasi-NC algorithms for perfect matchings on bipartite 
graphs \cite{FGT16} and
on general graphs \cite{ST17}.
Furthermore, derandomizing the corresponding algorithm for \emph{general} linear 
spaces of
matrices -- not necessarily those of the form $\cB_G$ or $\cA_G$ -- is now known 
as the symbolic determinant identity testing problem, and turns out to 
be of fundamental significance in complexity theory, as that would imply strong 
circuit lower bounds
which are considered to be beyond current techniques \cite{KI04,CIKK15}.

In the following, we shall call linear spaces of (alternating) matrices as
(alternating) matrix spaces. For a
field $\F$, we use $\M(s\times t, \F)$ to denote the linear space of all $s\times 
t$
matrices over $\F$, and write $\cB\leq \M(s\times t, 
\F)$ to denote that $\cB$ is a matrix
space in $\M(s\times t, \F)$. Let $\M(n, \F):=\M(n\times n, \F)$, and $\Lambda(n, 
\F)$ be
the linear space of all $n\times n$ alternating matrices over $\F$.


\paragraph{The bridge between shrunk subsets and shrunk subspaces.} For bipartite
graphs, a structure closely related to perfect matchings is the following. Given a
bipartite graph $G=(L\cup R, E)$ where $L=R=[n]$, we say that a subset $S\subseteq
L$ is a
shrunk\footnote{We call $S$ to be ``shrunk'' instead of ``shrinking'', as we
think of the bipartite graph $G$ as shrinking the set $S$.} subset of $G$, if
$|S|>|N(S)|$ where $N(S)$ is the set of neighbours of $S$ in
$R$. The celebrated Hall's marriage theorem \cite{Hal35} says that $G$ has a 
perfect matching if and only if it does not have a shrunk subset.

On the matrix space side, it is then natural to define the so-called shrunk
subspaces. Specifically, given a matrix space $\cB\leq \M(n, \F)$, a subspace
$U\leq \F^n$ is a shrunk subspace of $\cB$, if $\dim(U)>\dim(\cB(U))$ where
$\cB(U):=\langle \cup B(U) : B\in\cB\rangle$, and $B(U)$ denotes the image of $U$
under $B$.

As in the perfect matching case, a bipartite graph $G$ has a shrunk subset if and
only if $\cB_G$ has a shrunk subspace \cite{Lov89}. However, for general matrix 
spaces, the natural analogue
of Hall's theorem, namely a matrix space contains full-rank matrices if and only
if it has no shrunk subspaces, does not hold, as evidenced by the space of all
$3\times 3$ alternating matrices. (The only if direction
trivially holds,
though.) 
Therefore, the bridge between shrunk subsets and shrunk subspaces is
different 
from the one between perfect matchings and full-rank matrices. 

The problem of testing
whether a matrix space has a shrunk subspace arises naturally from
several mathematical and computational displines, including algebraic complexity,
non-commutative algebra, invariant theory, and analysis
\cite{GGOW16,IQS17,GGOW17}. Not surprisingly then, this problem has had several 
names. We adopt the \emph{non-commutative rank problem} which seems widely 
used now, and refer an interested reader to \cite{GGOW16,IQS17} for the 
origin of this name.

With all these motivations, the non-commutative rank problem recently
received considerable attention, and substantial progress has been made.
First raised by Cohn \cite{Cohn75} four
decades ago in the
study of free fields, it was more recently reached at by Mulmuley
in
the context of
derandomizing the Noether's Normalization Lemma \cite{Mul12,Mul17}, and also by
Hrube\v{s}
and
Wigderson in the context of non-commutative arithmetic circuits with divisions
\cite{HW15}. Only known to be in PSPACE before 2015 \cite{CR99}, this problem was
shown
to be in P over the rational number field \cite{GGOW16} and over any
field \cite{IQS17}.

The techniques supporting the solutions to the non-commutative rank problem are
reminiscent of the corresponding techniques for the perfect
matching problem on bipartite graphs. In \cite{GGOW16}, it is the scaling
algorithm
\cite{Sin64,LSW00}, generalized to the quantum operator setting \cite{Gurvits}. In
\cite{IQS17}, it is the classical augmenting path algorithm, generalized to the
matrix space setting \cite{IKQS15,IQS16}. Ingredients from invariant theory
are also crucial. For \cite{GGOW16}, Garg et al. needed
the exponential upper bound
on generating the ring of matrix semi-invariants \cite{derksen_bound}. For
\cite{IQS17}, Ivanyos et al. need the polynomial upper bound \cite{DM2}, which in
turn relies crucially on the regularity lemma developed in \cite{IQS16}.

\subsection{Between graphs and matrix spaces: a new bridge}

In this paper we lay the foundation for yet another bridge between graphs and 
matrix spaces. We focus
on undirected simple graphs, hence it is natural, as Tutte and Lov\'asz did with 
perfect
matchings, to work with alternating matrix spaces. 
We start from independent sets and vertex 
colorings, two central structures in graph theory with
numerous results from various motivations \cite{JT95,diestel}. By identifying
analogues of them in the alternating matrix space
setting, we arrive at isotropic spaces and isotropic decompositions, which we 
define now.
\begin{definition}\label{def:main}
Let $\cA\leq \Lambda(n, \F)$ be an alternating matrix space. A subspace $U\leq
\F^n$ is an \emph{isotropic space} of $\cA$, if for any $u, u'\in U$, and any $A\in
\cA$, we have $u^tAu'=0$. For $c\in \N$, an \emph{isotropic $c$-decomposition} of
$\cA$ is a
direct sum
decomposition of $\F^n$ into $c$ nonzero subspaces $U_1\oplus \dots\oplus U_c$,
where every $U_i$ is an isotropic space.
\end{definition}

Recall that for a graph $G=([n], E)$, an independent set of $G$ is a subset
$S\subseteq
[n]$ such that for any $i, j\in S$, there is no edge from $E$ connecting these two
vertices.
A vertex
$c$-coloring of $G$ is a partition of the vertex set into $c$ independent sets.
Therefore, the definitions of isotropic spaces and
isotropic decompositions do
mimic those of independent sets and vertex colorings.
It is then natural to
introduce the following definitions and the corresponding algorithmic problems.
\begin{definition}
Let $\cA\leq \Lambda(n, \F)$. The \emph{isotropic number} of $\cA$, denoted as
$\alpha(\cA)$,
is the maximum $d\in\N$ such that $\cA$ has an isotropic space of dimension $d$. 
The \emph{isotropic
decomposition number}, denoted as $\chi(\cA)$, is the minimum $c\in \N$ such that
$\cA$ admits an isotropic $c$-decomposition.

Given $d\in \N$ and a linear basis of $\cA\leq \Lambda(n, \F)$, the \emph{maximum
isotropic space problem} asks to
decide whether $\alpha(\cA)\geq d$. Given $c\in \N$ and a linear basis of $\cA\leq
\Lambda(n, \F)$, the \emph{isotropic
$c$-decomposition problem}
asks to decide whether $\chi(\cA)\leq c$.
\end{definition}

Note that $\alpha(\cdot)$ and $\chi(\cdot)$ are used to denote the independent
number and the chromatic number of graphs \cite{diestel}, and these choices are
deliberate. Also note that for any $\cA\leq \Lambda(n, \F)$, we have
$\alpha(\cA)\geq 1$,
and $\chi(\cA)\leq n$. Indeed, due to the
alternating condition, any $\cA\leq \Lambda(n, \F)$ enjoys the property that any
$1$-dimensional subspace of $\F^n$ is an isotropic
space of $\cA$. It follows that any direct sum decomposition of $\F^n$ into $n$
dimension-$1$ subspaces is an isotropic $n$-decomposition of $\cA$. This property
corresponds
nicely to that for any undirected simple
graph, every
single vertex is an independent set. On the other hand, symmetric matrix
spaces do not satisfy this property in general. Therefore, this small but
pleasant
coincidence
suggests that working with alternating matrix spaces is a natural choice in this
setting.

Our first result follows
what Lov\'asz did with perfect matchings, and provides a first indication on the
new connection. Recall that given a graph
$G=([n], E)$, we can associate an alternating matrix spaces $\cA_G\leq \Lambda(n,
\F)$, spanned by those
elementary alternating matrices
$A_{i, j}$ with $\{i, j\}\in E$.
\begin{theorem}\label{thm:reduction}
Let $G$ and $\cA_G$ be as above. Then we have
\begin{enumerate}
\item $G$ has a size-$s$ independent set if and only if $\cA_G$ has a
dimension-$s$ isotropic space. In particular, $\alpha(G)=\alpha(\cA_G)$.
\item $G$ has a vertex $c$-coloring if and only if $\cA_G$ has an isotropic
$c$-decomposition. In particular, $\chi(G)=\chi(\cA_G)$.
\end{enumerate}
\end{theorem}
The proof is in Section~\ref{sec:reduction}. Theorem~\ref{thm:reduction}
demonstrates that the maximum isotropic space problem
and the isotropic decomposition problem are genuine generalizations\footnote{Note 
that for the maximum isotropic space problem
and the isotropic decomposition problem, we consider all alternating matrix
spaces, not necessarily of the form $\cA_G$ coming from a graph $G$.} of the
maximum independent set problem and the vertex $c$-coloring problem, respectively.
It also implies the following.
\begin{corollary}\label{cor:reduction}
The maximum isotropic space problem and the isotropic $3$-decomposition problem
for alternating matrix spaces
are NP-hard.
\end{corollary}

Emboldened by Theorem~\ref{thm:reduction}, we propose to view isotropic spaces 
and decompositions as linear algebraic
analogues of independent sets and vertex colorings, and study these two structures 
from the perspectives of graph theory and algorithms. This leads to natural and 
interesting mathematical and algorithmic problems, whose solutions 
bring together strategies, techniques, and results from several areas, including 
graph theory, algorithm design, computer algebra, and algebraic complexity. 
We regard these results as laying the foundation of a new bridge between graphs 
and alternating matrix spaces. 


While our investigation started with an analogy, isotropic spaces and 
decompositions are actually classical notions, with natural 
interpretations in group theory and
manifold theory. Therefore, some of our results have concrete applications 
to these areas. We also initiate a variant of this theory 
in quantum information, and find an interesting information theoretic 
interpretation of isotropic spaces in quantum error correction. These demonstrate 
the usefulness of 
this new 
bridge. 


Before we go on to detailed descriptions of our results, we set up some notation.
We use $\F_q$, $\Q$, $\R$, and $\C$ to denote the
finite field with $q$ elements,
the rational number field, the real number field, and the complex number field,
respectively. Elements in $\F^n$ are \emph{column} vectors. In algorithms,
subspaces of $\F^n$ and $\Lambda(n, \F)$ are represented by linear bases. We may
write $\Lambda(n, \F_q)$ as $\Lambda(n, q)$ for convenience.  For more
details on the computation model, see Section~\ref{subsec:algo_model}.

We
also recall some
well-known graph-theoretic and/or algorithmic results \cite{diestel}, which will
be useful in
seeing the analogues.
\begin{enumerate}
\item Whether a graph is bipartite can be tested in deterministic polynomial time.
\item On bipartite graphs, the maximum independent set problem is in P.
\item Any $n$-vertex and $m$-edge graph has an independent set of size $\geq
\frac{n^2}{2m+n}$ \cite{Tur41}\footnote{This follows from Tur\'an's celebrated
result in extremal graph theory, which is usually
stated for cliques, and implies this by simply taking the complement
graph.}.
\item The number of maximal independent sets on an $n$-vertex graph is $\leq
3^{\frac{n}{3}}$ \cite{MM65}.
\item The chromatic number of an $n$-vertex graph can be computed in time
$(1+3^{\frac{1}{3}})^n\cdot \poly(n)$ \cite{Law76}\footnote{This classical result 
of
Lawler
was from the 1970's, and the
current status of the art is $2^n\cdot \poly(n)$ \cite{BHK09}.}.
\end{enumerate}
All the above results will be found to have natural correspondences
in the alternating matrix space setting. The reader may find some fun in trying to
formulate the correspondences by him/herself.



\subsection{Linear algebraic analogues of bipartite testing and maximum 
independent set on bipartite graphs
}\label{subsec:2-decomp}

After
Corollary~\ref{cor:reduction}, the isotropic $2$-decomposition problem is of
particular interest, as the vertex $2$-coloring problem just asks whether a graph
is bipartite, which can be tested efficiently by breadth-first search.
A
moment's thought sugggests that a breadth-first search type idea seems not
applicable to the isotropic $2$-decomposition problem (see also
Appendix~\ref{app:breadth}).

Fortunately, it turns out
that this problem has been studied in computer algebra
over finite fields by Brooksbank, Maglione, and Wilson in \cite{BMW17}.
Their
strategy
can be readily applied to $\R$
and $\C$, using some ingredients from \cite{IQ18}.
\begin{theorem}[{\cite[Theorem 3.6]{BMW17}, \cite{IQ18}}]\label{thm:2-decomp}
The isotropic $2$-decomposition problem can be solved in randomized polynomial
time over $\F_q$ with $q$ odd, and in deterministic polynomial time over $\R$ and
$\C$.
Furthermore, over $\F_q$ with $q$ odd, the algorithm also outputs the linear bases
of the two
subspaces in an isotropic $2$-decomposition.
\end{theorem}

While a proof for $\F_q$ was already sketched in \cite{BMW17}, we still give an
exposition of this proof in Section~\ref{sec:star}. Besides indicating how to
handle
$\R$
and $\C$, we wish to give some reader a
flavor of how the so-called $*$-algebra technique, pioneered by Wilson
\cite{Wil09,Wil09b} in computer algebra, is applied to this setting.
This technique was recently shown to
be useful in polynomial identity
testing and multivariate cryptography 
\cite{IQ18}.

Theorem~\ref{thm:2-decomp} and its
proof reveal that the
isotropic spaces and the isotropic decompositions do have connections with, and
implications to, other disciplines, just like the case of non-commutative ranks.
Furthermore, a quantum variant of the theory can be developed, and the
corresponding isotropic $2$-decomposition problem can be solved efficiently
using quantum information theoretic techniques (see Section~\ref{sec:quantum}).
As mentioned above, the techniques used to solve the non-commutative rank problem
also have their roots in algebra \cite{IQS17} and quantum information
\cite{GGOW16}. Perhaps it is not so coincidental that techniques from these areas
are useful again.

In fact, the non-commutative rank problem arises naturally in our context.
Since a bipartite graph is just a graph
admitting a vertex $2$-coloring, it is natural to make the following definition.
\begin{definition}
An alternating matrix
space is \emph{bipartite}, if it admits an isotropic $2$-decomposition.
\end{definition}

A well-known fact in graph theory is that, on bipartite graphs, the maximum
independent set problem can
be solved in deterministic polynomial time, through a reduction to the minimum
vertex cover problem. The latter problem is equivalent to the maximum matching
problem via K\H{o}nig's
theorem. 
It is then interesting to examine whether bipartite
alternating matrix spaces admit an efficient algorithm for the maximum isotropic
space problem. It turns out that the isotropic number of a bipartite alternating
matrix space is closely related to the non-commutative rank of some matrix space, 
as we shall see now.

We have mentioned the decision version of the non-commutative rank problem in
Section~\ref{subsec:known_connection}. We now define the non-commutative rank in a
slightly more general setting. Given $\cB\leq \M(s\times t, \F)$, its
non-commutative rank is $\ncrk(\cB):=s+t-\max\{\dim(U)+\dim(V) :
\forall u\in U, v\in V, B\in \cB, u^tBv=0\}$ \cite{FR04}. Note that the recent 
works \cite{GGOW16,IQS17} used a slightly different formulation in the setting 
$s=t$.

Given a bipartite $\cA\leq\Lambda(n, \F)$, up to isometry (i.e. the action of
$T\in \GL(n, \F)$ sending $\cA$ to $T^t\cA T:=\{T^tAT : A\in \cA\}$), every
$A\in \cA$ is of the form $\begin{bmatrix} \zerovec & B \\ -B^t &
\zerovec\end{bmatrix}$, where $B$ is of size $s\times t$ (see
Section~\ref{sec:basic}). Let $\cB\leq \M(s\times
t, \F)$ be the space of such $B$ arising from some $A\in \cA$. Then we have:
\begin{theorem}\label{thm:maximum_bipartite}
Let $\cA\leq \Lambda(n, \F)$ and $\cB\leq \M(s\times t, \F)$ be from above. Then
$\alpha(\cA)=n-\ncrk(\cB)$.
\end{theorem}

Thanks to the solution of the non-commutative rank problem over any
field \cite{IQS17}, and Theorem~\ref{thm:2-decomp} in the case of $\F_q$ with odd
$q$, we have
\begin{corollary}\label{cor:maximum_bipartite}
The maximum isotropic space 
problem for bipartite alternating matrix spaces of size $n\times n$ over $\F_q$,
$q$ odd, can be
solved in randomized $\poly(n, \log q)$ time.
\end{corollary}
The proofs of Theorem~\ref{thm:maximum_bipartite} and
Corollary~\ref{cor:maximum_bipartite} are in Section~\ref{sec:maximum_bipartite}.
In some sense, the non-commutative rank may be considered as corresponding to the
minimum vertex
cover size in
the bipartite alternating matrix space setting. However, unlike in the graph case,
where
the relation between independent sets and vertex covers is so straightforward,
the proof of Theorem~\ref{thm:maximum_bipartite} requires some twists, because of
the ``flexibility'' of vectors and matrices.


Having seen the implication of the non-commutative rank problem to our setting,
let us examine the following mathematical problem that arises naturally in our
context, whose solution turns out to
come from algebraic geometry. Again, let us trace back to the
graph setting, and consider $\alpha(n, m):=\min\{\alpha(G) : G\text{ a graph with
}n\text{ vertices and }m\text{ edges}\}$, where $\alpha(G)$
is
the independence
number. A celebrated result of Tur\'an \cite{Tur41}
in extremal graph theory implies that
\begin{equation}\label{eq:turan}
\alpha(n, m)\leq \Big\lceil
\frac{n^2}{2m+n} \Big\rceil.
\end{equation}
Turning to the alternating matrix space setting, it is natural to define
$\alpha(\F, n,
m):=\min\{\alpha(\cA) : \cA\leq \Lambda(n, \F), \dim(\cA)= m\}$. This quantity
has
been studied by Buhler, Gupta, and Harris \cite{BGH87} in relation to abelian
subgroups of $p$-groups
\cite{Bur13,Alp65}. The main result of \cite{BGH87}, proved using algebraic
geometric techniques, is as follows: 
for any $m>1$, we have 
\begin{equation}\label{eq:bgh}
\alpha(\F, n, m)\leq \Big\lfloor \frac{m+2n}{m+2}\Big\rfloor,
\end{equation}
where the equality is attainable over algebraically closed fields\footnote{While 
in \cite{BGH87} the main result was stated for fields of characteristic $\neq 2$, 
the 
proof, at least the inequality, works for any characteristic.}. 
This inequality was also 
obtained earlier by Ol'shanskii \cite{Ols78}.
Comparing Equations~\ref{eq:turan} and~\ref{eq:bgh}, we see that $\alpha(n, m)$
and $\alpha(\F, n, m)$ behave quite differently. For example, by
Equation~\ref{eq:turan}, every graph with $n$ vertices and $2n$ edges
has an independent set of size at least $n/5$. On the other hand, by
Equation~\ref{eq:bgh}, there exists a dimension-$2n$ alternating matrix space in
$\Lambda(n, \F)$ with no isotropic space of dimension $\geq 2$. 


Motivated by the discussion in the last paragraph, we study the
algorithmic problem of deciding whether there exists an isotropic space of
dimension $\geq 2$ for $\cA\leq \Lambda(n, \F)$. This is equivalent to ask whether
$\cA$ has an isotropic $(n-1)$-decomposition. Note that the corresponding problem
on graphs is trivial, as
a graph has an
independent set of size $\geq 2$ if and only if it is not the complete graph. It
turns out that over $\Q$, this problem is substantially more difficult.
\begin{theorem}\label{thm:iso_dim_2}
Over $\Q$, assuming the generalized Riemann hypothesis, there is a randomized
polynomial-time
reduction from deciding quadratic residuosity modulo squarefree composite
numbers
to the problem
of deciding whether an
alternating matrix space has an isotropic space of dimension $\geq 2$.
%
\end{theorem}
The proof is in Section~\ref{sec:iso_dim_2}. It relies crucially on
R\'onyai's fundamental work on computing algebra structures \cite{Ron87}. One
ingredient here is the introduction of the existential singularity problem
for matrix spaces, which turns out to have
rich connections to several mathematical disciplines (see
Problem~\ref{prob:exist_sing} and Section~\ref{subsec:exist_sing}).




\subsection{Linear algebraic analogues of bounding the number of maximal 
independent sets and exact exponential algorithms for chromatic numbers
}\label{subsec:exact}

An independent set on a graph is maximal if it is not properly contained in some
other independent set. The study of maximal independent sets is a classical
demonstration
of
how graph theory and algorithm study are intertwined.

In the 1960's, Erd\H{o}s and Moser raised the question of bounding the number
of maximal independent sets on a graph. It was subsequently solved by Moon and
Moser \cite{MM65}, and alternative proofs have been found \cite{Wood11,Vat11}.
They show that
the number of maximal independent set of an
$n$-vertex graph is upper bounded by $3^{n/3}$,
and this bound is tight. (Some
refinement is required when $n$ is not a multiple of $3$.)
Since the 1970's, the problem of outputting all maximal independent sets or maximal
cliques received considerable attention
\cite{AM70,TIAS77,LLK80,JPY88}. One
application was
provided by
Lawler \cite{Law76}, who showed that the Moon-Moser bound together with dynamic
programming give an algorithm for computing the chromatic number of an $n$-vertex
graph in time $(1+\sqrt[3]{3})^n\cdot \poly(n)$. This algorithm was the
starting
point of exact exponential-time algorithms for chromatic numbers. Subsequent
improvements \cite{eppstein2003small,byskov2004enumerating,bjorklund2008exact}
lead to the breakthrough by Bj\"orklund, Husfeldt, and Koivisto, who presented an
algorithm in time $2^n\cdot \poly(n)$ \cite{BHK09}.

Getting back to alternating matrix spaces, the natural correspondence would be
maximal isotropic
spaces. Formally, for an
alternating matrix space $\cA\leq \Lambda(n, \F)$, an isotropic space  is
\emph{maximal}, if there is no isotropic space properly containing it. We then ask
analogous questions over finite fields, namely upper bounding the number of maximal
isotropic spaces of $\cA\leq \Lambda(n, \F_q)$, and exact exponential-time
algorithms for computing the isotropic decomposition number $\chi(\cA)$.
Interestingly, on one
hand, these
problems demonstrate behaviours different from the combinatorial counterpart. On
the other hand, strategies and techniques from graph theory and algorithm design
do carry over, in a non-trivial way, to these problems. Again, such phenomena have
been witnessed in the non-commutative rank problem, and it is interesting to see
these happening in this context. Furthermore, our result on the number of maximal 
isotropic spaces has a direct application to group theory, as we will see in 
Section~\ref{subsec:inter1}.

We now describe our results in more details. To start with, we note that, as in
the graph setting, an easy greedy algorithm outputs one maximal isotropic space
(see Proposition~\ref{prop:maximal_is}). We then consider the number of maximal
isotropic spaces for
alternating
matrix spaces in $\Lambda(n, q)$, analogously as done by Moon and Moser for
graphs \cite{MM65}. A trivial upper bound is the number of all
subspaces of $\F_q^n$. This number, $q^{\frac{1}{4}n^2+\Theta(n)}$, is well-known
and classical (see Fact~\ref{fact:gaussian}). Any alternating matrix space spanned
by a single
full-rank
alternating matrix has $q^{\frac{1}{8}n^2+\Theta(n)}$ maximal isotropic spaces,
providing a lower bound. This is also classical but perhaps not that well-known
(see Proposition~\ref{prop:iso_bd}).
We show a non-trivial upper bound as follows.
\begin{theorem}\label{thm:iso_bound}
The number of maximal isotropic spaces of any $\cA\leq \Lambda(n, q)$ is upper
bounded by $q^{\frac{1}{6}n^2+O(n)}$.
\end{theorem}
The proof is in Section~\ref{sec:iso_bound}. We adapt
the proof
strategy of the upper bound on maximal independent sets by Wood \cite{Wood11}.
This requires analogues of certain graph-theoretic concepts such as degrees and
neighbours
in
the alternating matrix space setting, which have been developed in \cite{Qia18}.
It
works up to some point, but after that, we have to resort to certain linear
algebraic techniques. We leave closing the gap between
$q^{\frac{1}{8}n^2+\Omega(n)}$ and $q^{\frac{1}{6}n^2+O(n)}$ an interesting open
problem.

The proof of Theorem~\ref{thm:iso_bound} is constructive (see
Corollary~\ref{cor:iso_bound}), so we can enumerate
all maximal isotropic spaces in time $q^{\frac{1}{6}n^2+O(n)}$.
We then
consider the problem of computing the isotropic decomposition number for $\cA\leq
\Lambda(n,
q)$. A naive brute-force algorithm, namely enumerating all
direct sum decompositions of $\F_q^n$, runs in time $q^{n^2+O(n)}$.
Inspired by Lawler's strategy in
\cite{Law76}, we combine our Corollary~\ref{cor:iso_bound} with a
dynamic programming idea to obtain the following.
\begin{theorem}\label{thm:lawler}
The isotropic decomposition number of $\cA\leq \Lambda(n, q)$ can be
computed in
time $q^{\frac{5}{12}n^2+O(n)}$.
\end{theorem}
The proof is in Section~\ref{sec:lawler}.
An open question is whether the strategy in
\cite{BHK09} for chromatic numbers can be adapted to obtain an algorithm for
isotropic decomposition numbers in time
$q^{\frac{1}{4}n^2+O(n)}$. This is because the
number
of subspaces of $\F_q^n$ is $q^{\frac{1}{4}n^2+\Theta(n)}$, while the algorithm in 
\cite{BHK09} runs in time
$2^n\cdot \poly(n)$ where $2^n$ is the number of subsets of $[n]$.

\subsection{Complexity-theoretic upper bound over $\C$, and the dependence of the
independence number on the maximum degree}



We first consider complexity-theoretic upper bounds for the maximum isotropic
space problem and the isotropic $3$-decomposition
problem. Clearly, these problems are in NP over finite fields. Over $\C$, we have
the following, by resorting to a celebrated result of Koiran
on the Hilbert Nullstellensatz problem
\cite{Koi96}. The proof is in Section~\ref{sec:maximum_isC}.

\begin{proposition}\label{thm:maximum_isC}
Let $\cA\leq \Lambda(n, \C)$ be given by a linear basis consisting of integral
matrices. The maximum isotropic space problem and the isotropic $3$-decomposition
problem
are in PSPACE unconditionally, and in
PH assuming the generalized Riemann hypothesis.
\end{proposition}

Our next result has two diverse motivations.

The first motivation is from linear algebra. Given a single alternating matrix
$A$, its canonical
form suggests that $\langle A\rangle$ admits an isotropic $2$-decomposition (see
Section~\ref{sec:prel}). Given
a pair of alternating matrices $A_1$ and $A_2$, it is also known that $\langle
A_1, A_2\rangle$ admits an isotropic $2$-decomposition \cite{Sch76,GL84,GG07} (see
also
\cite[Lemma 3.7]{BMW17}). A natural question is what happens for
alternating
matrix
spaces of dimension $3$, or in general, any constant $c$.

The second motivation is from graph theory.
Given a graph $G=([n], E)$, let $\Delta(G)$ be the maximum degree over vertices 
of  $G$. It is
well-known that a simple greedy algorithm yields that $\chi(G)\leq \Delta(G)+1$
\cite[pp. 122]{diestel}. For $\cA\leq \Lambda(n, \F)$, the \emph{degree} of $v\in 
\F^n$ in $\cA$
can be
defined as $\deg_\cA(v):=\dim(\langle Av: A\in \cA\rangle)$ \cite{Qia18}. As
mentioned in Section~\ref{subsec:exact}, this
notion was already useful in the proof of Theorem~\ref{thm:iso_bound}. Let
$\Delta(\cA)=\max\{\deg_\cA(v) : v\in \F^n\}$. It is then natural to ask the 
relation between $\chi(\cA)$ and $\Delta(\cA)$ in analogy to the graph setting. 
This question is closely related to the one in the last paragraph, since
$\deg_\cA(v)\leq \dim(\cA)$ for any $v\in \F^n$, so $\Delta(\cA)\leq\dim(\cA)$. 

We now present the following result, also deduced from a greedy algorithm.

\begin{proposition}\label{prop:bounded_deg}
Let $\cA\leq \Lambda(n, \F)$. Then $\chi(\cA)\leq O(\Delta(\cA)\cdot \log n)$.
Furthermore, an isotropic $C$-decomposition with $C=O(\Delta(\cA)\cdot \log n)$
can be found in polynomial time.
\end{proposition}
The proof is in Section~\ref{sec:maximum_isC}. Note that this implies that when
$\dim(\cA)$ is a constant, then $\chi(\cA)\leq
O(\log n)$. We leave it an open problem for further improvement of the
bound in Proposition~\ref{prop:bounded_deg}.

\subsection{Applications of our results}\label{subsec:inter1}

After studying problems on 
alternating matrix spaces mostly by way of analogy
with graphs,
it is natural to ask whether some results have concrete applications. The 
answer is quite affirmative.

In this subsection, 
we first provide two applications to finite groups, one being computational, and 
the other being enumerative. As will be explained below, 
these applications are based 
on a family of finite groups, for which testing isomorphism has 
long been known to be difficult, and is becoming more urgent in 
light of 
Babai's recent breakthrough on graph isomorphism \cite{Bab16}. 

We then describe a variant of our
theory in the context of quantum information theory, and present an 
information theoretic interpretation of isotropic spaces in the context of quantum 
error correction. 
In Section~\ref{app:origin},
we also present the connections to 
manifold theory, and mention a potential application.
All these suggest
that our results could be of interest to group theorists, quantum information
theorists, and geometers, in particular to those who work on the 
computational aspects of these
disciplines. 

Of course,
these applications and connections are not surprising to readers
in these fields, because alternating bilinear maps, and therefore 
alternating matrix
tuples and spaces\footnote{For the
relations
among alternating matrix tuples and spaces, and alternating bilinear maps, see
Sec.~\ref{sec:prel} and~\ref{app:spaces_maps}.}, naturally arise in 
group theory via the commutator bracket, and in manifold
geometry via the cup product in cohomology.
Therefore, certain isotropic spaces and decompositions have
natural group-theoretic or geometric interpretations (see
Section~\ref{app:origin}).

On the other hand, these applications may look exotic to some other readers, as
they will be stated
purely in group theoretic or quantum information theoretic terms. This is natural 
and expected, after a bridge is built. Indeed, the present bridge enables 
us 
to transfer problems, techniques, and results in graph theory and algorithm
study, to other mathematical and computational disciplines which otherwise seem
barely related to graph theory. 



We now describe the first application to finite groups, more specifically, to 
computing with matrix groups over finite fields. 
Matrix groups over
finite fields given by generators form an important model of computing with finite
groups.
In theoretical computer science, the study of this model led to the inventions of
the
black-box group model
by Babai and
Szemer\'edi \cite{BS84}, and the Arthur-Merlin class by Babai \cite{Bab85}. Though 
some 
algorithms with worst-case analyses can
be found in \cite{Bea95}, even the very basic membership testing problem
was
only recently known to be solvable in randomized polynomial time under a
number-theoretic oracle \cite{BBS09}. 
%
%

Overall, our knowledge about this model is rather limited, and 
many
questions await investigations. One interesting problem is to compute an abelian
subgroup of the largest size. Large abelian subgroups, besides motivations from
computational group theory \cite{Ros11}, are useful in controlling the character
degrees of the group, which in turn are useful in the group-theoretic
approach for fast matrix multiplication \cite{CKSU05}. As a consequence of
Corollary~\ref{cor:reduction}, we have the following result, whose proof is in
Section~\ref{sec:application}.
\begin{theorem}\label{thm:large_abelian}
Let $p$ be an odd prime. Given a matrix group $G$ over $\F_p$, and $s\in\N$, 
deciding whether $G$ has an abelian subgroup of order $\geq s$ is NP-hard.
\end{theorem}

The proof
of Theorem~\ref{thm:large_abelian} relies on the connection between alternating
matrix spaces over $\F_p$, and $p$-groups of class $2$ and exponent
$p$ for odd $p$, via Baer's correspondence \cite{Bae38} (see 
Section~\ref{app:origin}). 


It has long been known that
$p$-groups of class 2 and exponent $p$ form a
bottleneck for testing isomorphism of finite groups. To solve the group 
isomorphism problem in time
polynomial in
the group order is a long-standing open problem \cite{Mil78}. This 
problem is
becoming more prominent in light of Babai's breakthrough on graph isomorphism
\cite{Bab16}, as
Babai indicated the group isomorphism problem as a key bottleneck
to put graph isomorphism in P \cite[arXiv version, Section 13.2]{Bab16}. Some 
interesting progress on testing isomorphism of such $p$-groups was recently made 
by 
utilizing the 
connection to alternating matrix spaces \cite{LQ17}.



%

Let us turn to the second application to finite groups. 
The question of bounding the number of maximal
abelian subgroups has been considered for various group families 
\cite{Dix71,Woo89,Ant87,Vdo04}, but to the best of our knowledge, there had been 
no results on  
this question for 
$p$-groups of class 
$2$ and exponent $p$. Let
$P$ be such a group, so that the center $Z(P)\cong \Z_p^m$ and the 
central quotient 
$P/Z(P)\cong \Z_p^n$. 
The number of maximal
abelian subgroups is upper bounded trivially by $p^{\frac{1}{4}n^2+O(n)}$,
the number of subgroups of $\Z_p^n$. 
Our Theorem~\ref{thm:iso_bound} then provides the following improvement, whose 
proof is in
Section~\ref{sec:application}.
\begin{theorem}\label{thm:maximal_abel}
Let $P$ be as above. Then the number of maximal abelian subgroups of $P$ is upper
bounded by $p^{\frac{1}{6}n^2+O(n)}$.
\end{theorem}
Recall that the
proof of Theorem~\ref{thm:iso_bound}
starts by following the strategy of Wood's proof \cite{Wood11} of bounding the
number of maximal independent sets on a graph. We view this as an interesting and 
somewhat
unexpected example of transferring techniques from graph theory to group
theory.

We also initiate a quantum variant of the theory in Section~\ref{sec:quantum}.
There, the objects are a special
type of quantum channels, and isotropic spaces and isotropic decompositions are
defined on the Kraus operators of such channels. Furthermore, we require the
isotropic decompositions to be orthogonal. One can then transform classical
connected graphs into such channels, and prove an analogue of
Theorem~\ref{thm:reduction}. More surprisingly, we also obtain an efficient
isotropic $2$-decomposition algorithm, as an analogue of 
Theorem~\ref{thm:2-decomp}, by resorting to the recent development on
the periodicity of quantum Markov chains \cite{guan2018decomposition}.

We then present an information theoretic 
interpretation for isotropic spaces in the context of quantum error correction. 
Briefly speaking, from the viewpoint of 
certain natural 
generalizations of 
quantum 
gate fidelities\cite{nielsen2002quantum}, isotropic spaces can be viewed as the 
opposite 
structure of noiseless 
subspaces (Proposition~\ref{prop:quantum}), which have been studied 
intensively in 
quantum error 
correction~\cite{KL97,lidar2012review}. 
Indeed, noiseless subspaces are shelters for the information 
residing in them under quantum noise, while 
the information in an isotropic space would be completely destroyed by quantum noise.


Let us conclude this subsection with a remark on these applications. After 
building 
a bridge, we expect it to serve as a two-way street between the two sides. 
However, in reality there is usually more traffic in one direction than the other. 
For example, the traffic between
perfect matchings and full-rank matrices mostly goes from the algebra side to the
combinatorial side, e.g., the randomized NC algorithm for perfect matchings 
\cite{Lov79}. The traffic between  shrunk
subsets and
shrunk
subspaces mostly goes in the other direction, e.g., linear algebraic analogues of 
augmenting paths \cite{IKQS15} and scaling \cite{Gurvits}.  
So far, our applications in this work mostly go in the direction from
combinatorics to algebra, following the pattern of the
shrunk subset vs. shrunk subspace case. It will be very interesting to explore 
implications in the other direction in the future.

\subsection{Outlook}

\paragraph{Summary of our contributions.} The concepts of isotropic spaces and
isotropic
decompositions for alternating bilinear maps are classical, with natural 
interpretations in group theory and
manifold theory. 
Our key new insight is
that
they can be viewed and studied as linear algebraic analogues of independent sets
and vertex colorings. 
This insight leads us to study algorithmic and mathematical problems about 
isotropic 
spaces and isotropic decompositions, by drawing inspirations from results and 
techniques from graph theory and algorithm study. 
The techniques used to address the problems range from combinatorics, to algebra, 
and to quantum information. 

We believe that this
investigation is fruitful, for the following reasons.
\begin{enumerate}
\item First, it discloses new algorithmic and mathematical questions. For example, 
in Section~\ref{subsec:exact} we proposed and studied upper 
bounding 
the number of maximal isotropic spaces, and exact exponential-time 
algorithms for isotropic decomposition numbers  over finite fields. 
\item Second, the results obtained have concrete applications to other 
mathematical and computational disciplines. For example, in 
Section~\ref{subsec:inter1}, we described the applications of our results to 
finite groups and quantum information.
\item Third, it sheds new lights on known results from different research 
directions. For example, in 
Section~\ref{subsec:2-decomp} we compared Tur{\'a}n's extremal 
graph result with 
Buhler, Gupta, and Harris' algebraic geometric result.
\end{enumerate}

This investigation then lays the foundation for yet another bridge between 
graphs and alternating matrix spaces, adding to the classical ones established by 
Tutte and Lov\'asz.

\paragraph{Open ends.}
Several
interesting open problems have been mentioned before, and here we give a summary
and propose some new ones.
\begin{enumerate}
\item By Theorem~\ref{thm:iso_bound}, the number of maximal isotropic spaces of
$\cA\leq \Lambda(n, q)$ is upper bounded by $f(n, q)=q^{\frac{1}{6}n^2+O(n)}$.
There exists an alternating matrix space with $g(n,
q)=q^{\frac{1}{8}n^2+\Omega(n)}$ many maximal isotropic spaces (see
Section~\ref{subsec:exact}). Either improve the current upper bound $f(n,q)$, or
construct an alternating
matrix
space
with more than $g(n, q)$ maximal isotropic spaces. Note that resolving this 
problem would lead to a sharp bound on the number of maximal abelian subgroups of 
$p$-groups of class $2$ and exponent $p$.
\item Improve the exact exponential-time algorithm for computing the isotropic
decomposition number for $\cA\leq \Lambda(n, q)$ in Theorem~\ref{thm:lawler}. An
interesting question is
whether the strategy in \cite{BHK09} can be adapted here. The results in
\cite{BHK+16} should be useful in this context.
\item Despite Theorem~\ref{thm:iso_dim_2}, the complexities of deciding whether an
alternating matrix space has an isotropic
space of dimension $\geq 2$ are not clear over various fields. Even over $\Q$, our
proof for Theorem~\ref{thm:iso_dim_2} relies on a special case of the underlying
existential singularity problem for matrix spaces (see
Section~\ref{subsec:exist_sing}), so it is
left open even for the general case of
that problem over $\Q$.
\item Investigate the behaviours of the isotropic and isotropic decomposition
numbers in the linear algebraic Erd\H{o}s-R\'enyi model
\cite{LQ17,bollobas2001random}.
\item Improve the dependence of the isotropic decomposition number on the maximum
degree, or the dimension of the alternating matrix space
(see Proposition~\ref{prop:bounded_deg}). Note that this problem has motivations 
from classical geometry (see the discussions before 
Proposition~\ref{prop:bounded_deg}).
\end{enumerate}

\paragraph{The structure of the paper.}
We present
certain preliminaries in Section~\ref{sec:prel}. Then in Section~\ref{sec:basic},
we collect some basic facts and properties about isotropic spaces and
decompositions, including their meanings in group theory and 
manifold geometry in Section~\ref{app:origin}. We
then prove Theorem~\ref{thm:reduction} in Section~\ref{sec:reduction}, which is
the basis
connecting the graph-theoretic
structures and those structures on alternating matrix spaces.
We then prove all the main results
mentioned above in the following sections. (We have mentioned the corresponding
section numbers when describing those results.) An appendix then follows,
containing
some
background material and discussions on certain conceptual questions.


\section{Preliminaries}\label{sec:prel}

\paragraph{Notation.} For $n\in \N$, $[n]:=\{1, \dots, n\}$. We use $\uplus$ for
disjoint union of sets. The base of the logarithm is $2$ unless otherwise stated.

Let $\F$ be a field. We use $\F^n$ to denote the vector
space of
\emph{column} vectors of
length $n$ over
$\F$.
The standard basis of $\F^n$ consists of vectors $e_1, \dots, e_n$, where
$e_i$ is
the vector with the $i$th entry being $1$, and other entries being $0$.
The linear span of several vectors or matrices is denoted by $\langle \cdot
\rangle$.

For $n, d\in \N$, let $\M(n\times d, \F)$ be the linear
space of $n\times d$ matrices over $\F$, and $\GL(n\times d, \F)$ the set of
$n\times d$ matrices over $\F$ of rank $\min(n, d)$. We also let $\M(n,
\F):=\M(n\times n, \F)$, and $\GL(n, \F):=\GL(n\times n, \F)$. Dimension-$d$
subspaces of
$\F^n$ will
be understood as represented by elements from $\GL(n\times d, \F)$. Given $A\in
\M(n\times d, \F)$, the transpose of $A$ is denoted by $A^t\in \M(d\times n, \F)$.
For convenience, we sometimes write a vector $v$ in $\F^n$ as $v=(v_1, \dots,
v_n)^t$.
Let $\K/\F$ be a quadratic extension, and let $\overline\alpha$ denote the image
of $\alpha\in \K$ under the quadratic involution. For a matrix $A\in \M(n\times d,
\K)$, $A^\dagger$ denotes the conjugate transpose of $A$.

Depending on
the context, $\zerovec$ may denote either the zero
space, a zero vector, or a zero matrix. The identity matrix in $\M(n, \F)$ is
denoted by $I_n$; we may drop the subscript $n$ when it is understood from the
context.
Given a matrix $A\in \M(n, \F)$, its
kernel and image are denoted by
$\ker(A)$ and $\im(A)$, respectively. For $U\leq \F^n$, the image of $U$ under $A$
is denoted by $A(U)$.

\paragraph{Linear algebra.} Given $U\leq \F^n$, a complementary subspace, or just
a complement, is some $V\leq \F^n$ such that $V\cap U=\zerovec$, and $\langle
U\cup V\rangle=\F^n$. Note that complement subspaces of $U$ are not unique.
Indeed, the
number of complements of a dimension-$d$ subspace $U\leq \F_q^n$ is $q^{d(n-d)}$.
The
space orthogonal to $U$ is $\{v\in \F^n : \forall
u\in U, v^tu=0\}$. (Over $\C$, the conjugate transpose is used.) Note that the
space orthogonal to $U$ is not necessarily a
complement to $U$.

\paragraph{On matrix spaces.} Given a matrix space $\cA\leq \M(s\times t, \F)$,
the image of $U\leq \F^t$
under $\cA$ is $\cA(U):=\langle \cup_{A\in \cA}A(U)\rangle$. The dimension of
$\cA$ is denoted by $\dim(\cA)$. The (maximum) rank of $\cA$ is
$\rk(\cA):=\max\{\rk(A) : A\in \cA\}$. Let $\cB\leq \M(s\times t, \F)$ be another
matrix space. We say that $\cA$ and $\cB$ are \emph{equivalent}, if there exist
$C\in\GL(s, \F)$ and $D\in \GL(t, \F)$, such that $\cA=C\cB D:=\{CBD : B\in
\cB\}$. When working with $\cA$, an equivalence transformation is meant to left
multiply $\cA$ with some $C\in \GL(s, \F)$ and right multiply it with some $D\in
\GL(s, \F)$.

\paragraph{On alternating matrices.} Let $A, B\in \Lambda(n, \F)$. We say
that $A$ and $B$ are \emph{isometric}, if
there
exists $T\in \GL(n, \F)$, such that $A=T^tBT$.
Given a dimension-$d$ $U\leq \F^n$
represented by $T\in \GL(n\times d, \F)$, the \emph{restriction} of $A$ to $U$ by
$T$, denoted as $A|_{U, T}$, is
$T^tAT\in \Lambda(d, \F)$. The
radical of $A$ is the subspace $\{u\in \F^n : \forall v\in \F^n,
v^tAu=0\}$, which is just $\ker(A)$. The rank of $A$ is always even. If
$\rk(A)=2r$, then $A$ is isometric to $\begin{bmatrix} \zerovec & I_r & \zerovec
\\  -I_r & \zerovec & \zerovec \\ \zerovec & \zerovec & \zerovec\end{bmatrix}$
(see e.g. \cite[Chap. XV, Sec. 8]{Lang}). We say that $A\in \Lambda(n, \F)$ is
\emph{non-degenerate}, if $A$ is full-rank (so $n$ is even).

\paragraph{On alternating matrix spaces.}Let $\cA, \cB\leq \Lambda(n, \F)$.  We
say that $\cA$ and $\cB$ are
\emph{isometric}, if
there exists $T\in \GL(n, \F)$, such that $\cA=T^t\cB T:=\{T^tBT : B \in \cB\}$.
Given a
dimension-$d$ $U\leq
\F^n$ represented by $T\in \GL(n\times d, \F)$, the \emph{restriction} of $\cA$ on
$U$
via $T$ is $\cA|_{U, T}:=\{T^tAT : A\in \cA\}\leq \Lambda(d, \F)$. When it does
not cause confusion, we may not write $T$ explicitly, and just say the restriction
of $\cA$ to $U$, denoted by $\cA|_U$.
This
corresponds to the concept of induced subgraphs in graph theory. Indeed, we see
that $U$ is an isotropic space if and only if $\cA|_U$ is the zero (alternating
matrix) space. Given
$v\in
\F^n$, the \emph{radical of $v$ in $\cA$}, denoted
as $\rad_\cA(v)$, is $\{u\in \F^n : \forall A\in \cA, u^tAv=0\}$ which is a
subspace of $\F^n$. Elements in $\rad_\cA(v)$ correspond to non-neighbours in
graph theory.
The
\emph{codegree} of $v$ in $\cA$, denoted as $\codeg_\cA(v)$, is
$\dim(\rad_\cA(v))$. Note
that $\codeg_\cA(v)\geq 1$ for nonzero $v$, as $v\in \rad_\cA(v)$. The
\emph{degree} of $v$ in
$\cA$ is $\deg_\cA(v):=n-\codeg_\cA(v)$. More generally, for $U\leq \F^n$,
$\rad_\cA(U)=\{v\in \F^n : \forall u\in U, \forall A\in \cA, u^tAv=0\}$.
When
$\cA$ is
clear from the
context, we may drop the subscript $\cA$ in $\rad_\cA(v)$, $\codeg_\cA(v)$,
$\deg_\cA(v)$, etc.. It is easy to see that for any $v\in \rad(U)$, we have
$U\leq \rad(v)$, or in other words, $U\leq \rad(\rad(U))$.

A vector
$v\in \F^n$ is called \emph{isolated} in $\cA$, if
for
any $A\in \cA$, $Av=\zerovec$, which is equivalent to say that $\deg(v)=0$. This
corresponds to the concept of isolated
vertices in graph theory. The \emph{radical} of $\cA$, $\rad(\cA)$, is the
subspace
of $\F^n$ consisting of all isolated vectors. We say that $\cA$ is
\emph{non-degenerate},
if $\rad(\cA)=\zerovec$, and \emph{degenerate} otherwise. If $\cA\leq \Lambda(n,
\F)$ is degenerate with $\dim(\rad(\cA))=d>0$, then $\cA$ is isometric to $\cA'$
where each $A\in \cA'$ is of the form $\begin{bmatrix}
A' & \zerovec \\
\zerovec & \zerovec
\end{bmatrix}$, where $A'\in \Lambda(n-d, \F)$.

\paragraph{Sets, tuples, and spaces.} Let $\cA\leq \Lambda(n, \F)$ be given by a
linear basis $A_1, \dots, A_m\in \Lambda(n, \F)$. We can collect them as a set
$\sA=\{A_1, \dots, A_m\}\subseteq \Lambda(n, \F)$. Sometimes it is also useful to
impose an order on them, and form a tuple $\bA=(A_1, \dots, A_m)\in \Lambda(n,
\F)^m$. We shall use calligraphic fonts for spaces, bold fonts for tuples, and
sans serif fonts for sets.

Suppose $\cA\leq \Lambda(n,
\F)$ is given by a linear basis $A_1, \dots, A_m\in \Lambda(n, \F)$. Then it is
clear that, given $U\leq \F^n$, for any $u, u'\in U$ and any $A\in \cA$,
$u^tAu'=0$, if and only if for any $u, u'\in U$ and any $i\in[m]$, $u^tA_iu'=0$.
We therefore can define isotropic spaces, and isotropic decompositions, for sets
or tuples of alternating
matrices. In particular, since alternating matrix tuples represent alternating
bilinear maps naturally (see Appendix~\ref{app:spaces_maps}), this observation
suggests that isotropic spaces and decompositions for alternating matrix spaces
and for alternating bilinear maps are basically the same object. Furthermore,
many, though not all, concepts introduced in Section~\ref{sec:prel} about
alternating matrix spaces can be translated natually to alternating bilinear maps,
including degrees, degeneracy, radicals, etc.. (Indeed, some notions there are
actually borrowed from alternating bilinear maps.) On the other hand, note that
the maximum rank in $\cA$ is more natually associated with the space perspective.
More discussions on the relation between these two notions are in
Appendix~\ref{app:spaces_maps}.

Therefore, in the context of isotropic spaces and decompositions, the choices
between
alternating matrix spaces and tuples usually do not matter much. The tuple
perspective is more natural from the algorithm perspective, because the input of
an alternating matrix space to an algorithm is usually an ordered basis. The space
perspective is more natural for forming the analogy with graphs, and more
naturally allows for some more notions including the maximum rank, which is
important in e.g. the proof of Theorem~\ref{thm:iso_dim_2}. The set perspective
will be used in the quantum variant of the theory in Section~\ref{sec:quantum}.
Therefore, it is best
to keep all three perspectives in mind, and see how they fit into our problems.

\subsection{Computational models}\label{subsec:algo_model}


We will work with two computational models, depending on the problems. The first
model may be called the exact model; see e.g. \cite{lovasz1986algorithmic}. This
is the
model to work with, if field
extensions are unavoidable. In this model, input matrices or vectors are over a
field $\E$ where $\E$ is a finite field extension over its prime field $\F$, so
$\F$ is either a field of prime order or $\Q$. Suppose $\dim_\F(\E)=d$. Then $\E$
is an extension of $\F$ by a single generating element $\alpha$. We represent
$\alpha$ by the minimal polynomial of $\alpha$ over $\F$, and an isolating
interval for $\alpha$ in the case of $\R$, or an isolating rectangle for $\alpha$
in the case of $\C$. Note that from this representation, one can approximate the
numerical value of $\alpha$ arbitrarily closely. When we say that we work over
$\R$ or $\C$, the input is
given as over some number field $\E$ in $\R$ or $\C$. The algorithm is allowed to
work with extension fields of
$\E$ in $\R$ or $\C$, as long as the extension degrees are polynomially bounded.

The second model may be called the arithmetic model. In this model, only basic
field operations are performed, and the issue of working with different field
extensions does not arise. Still, over number fields we will be concerned with the
bit complexities, though it is possible that we may be able to only bound the
number of arithmetic
steps, but not the bit complexities.

We shall mostly work with $\F_q$, $\Q$, $\R$, and $\C$ in this article, though
some results extend to number fields naturally. Sometimes,
we make further restrictions like requiring $q$ to be odd, or the input to be
integral.

\section{Basic facts and properties}\label{sec:basic}

In this section we collect some basic results about isotropic spaces and isotropic
decompositions.

\paragraph{On the definitions.} The following is a somewhat more intuitive
definition of these two notions.

Let $\cA\leq\Lambda(n, \F)$, and let $U\leq \F^n$ be an isotropic space.
Suppose $d=\dim(U)$. Then form a change of basis matrix $T\in \GL(n, \F)$, such
that its first
$d$ columns form a basis of $U$, and the rest columns together with the first $d$
ones span $\F^n$. Then for any $A\in \cA$, we see that $T^tAT$ is of the form
$\begin{bmatrix}\zerovec & B \\ -B^t & C\end{bmatrix}$ where $\zerovec$ is of size
$d\times d$. It is not hard to see that $\cA$ has a dimension-$d$ isotropic space
if and only if there exists such a change of basis matrix $T$.

Similarly, $\cA$ has an isotropic $c$-decomposition, if and only if there exist
$T\in \GL(n, \F)$, and $d_1, \dots, d_c\in \Z^+$ with $\sum_i d_i=n$, such that
for any $A\in \cA$,
$T^tAT=\begin{bmatrix}\zerovec &
A_{1,2} & \dots & A_{1,c} \\ -A_{1,2}^t & \zerovec & \dots & A_{2,c} \\ \vdots &
\vdots & \ddots & \vdots \\ -A_{1,c}^t & -A_{2,c}^t & \dots &
\zerovec\end{bmatrix}$, where the $i$th $\zerovec$ on the diagonal is of size
$d_i$.


\paragraph{Computing the radical.} For many problems about isotropic spaces, given
a degenerate $\cA\leq
\Lambda(n, \F)$,
usually it is possible to reduce to the non-degenerate case. This is facilitated
by the fact that the radical of $\cA$ is easy to compute.

\begin{observation}
Suppose $\cA=\langle A_1, \dots, A_m\rangle\leq \Lambda(n, \F)$ is given by a
linear basis. There is a polynomial-time algorithm that computes a linear basis of
$\rad(\cA)$.
\end{observation}
\begin{proof}
Observe that $\rad(\cA)=\cap_{i=1}^n\ker(A_i)$. Computing
$\ker(A_i)$ and the intersection of $\ker(A_i)$'s are standard linear algebraic
tasks that can be performed as stated.
\end{proof}


\paragraph{Isotropic spaces and radicals of subspaces.} Let $\cA\leq \Lambda(n,
\F)$.
Recall that for $U\leq \F^n$, we defined $\rad_\cA(U)=\{v\in \F^n : \forall u\in
U, v^tAu=0
\}$. The following observation is immediate.

\begin{observation}\label{obs:isotropic}
Let $U\leq \F^n$ and $\cA\leq \Lambda(n, \F)$. Then we have the following.
\begin{enumerate}
\item $U$ is an isotropic space of $\cA$ if
and only if $U\subseteq \rad(U)$.
\item $U$ is a maximal isotropic space of $\cA$ if and only if $U=\rad(U)$.
\end{enumerate}
\end{observation}

An easy application of Observation~\ref{obs:isotropic} gives a greedy
algorithm for computing one maximal isotropic space.
\begin{proposition}\label{prop:maximal_is}
Given a matrix space $\cA\leq \Lambda(n, \F)$, a maximal isotropic space can be
computed in
polynomially many arithmetic steps.
\end{proposition}
\begin{proof} 
We first present the algorithm.
Recall that $e_i$ is the $i$th standard basis vector of $\F^n$.
\begin{enumerate}
\item Let $U=\langle e_1\rangle$.
\item While $U\subsetneq \rad(U)$:
\begin{enumerate}
\item Take any $u\in \rad(U)\setminus U$.
\item $U\leftarrow \langle U, u\rangle$.
\end{enumerate}
\item Output $U$.
\end{enumerate}

To see the correctness, note that in Step (2.a), by the choice of $u$, we have
$U\subseteq \rad(\langle U, u\rangle)$ and $u\in \rad(\langle U, u\rangle)$, so
$\langle U, u\rangle$ is an isotropic space by Observation~\ref{obs:isotropic}
(1). In Step (3), $U$ satisfies $U=\rad(U)$, so $U$ is maximal by
Observation~\ref{obs:isotropic} (2).

To see the running time, note that the while loop will be executed by at most $n$
times, since $\dim(U)$ increases by $1$ in each execution. Each step
involves
basic linear algebraic computations which require only polynomially many
arithmetic operations.
\end{proof}

Over $\R$ or $\C$, the bit sizes  in the above algorithm may blow up, at
least with
a
straightforward implementation, due to the iterative computations of the radicals.

\paragraph{On field extensions.} Let $\K$ be an extension field of $\F$. Given
$\cA\leq \Lambda(n, \F)$, we let $\cA_\K\leq \Lambda(n, \K)$ be the alternating
matrix spaces when we allow linear combinations over $\K$, or in other words,
$\cA_\K=\cA\otimes_\F\K$. We may write $\cA$ as $\cA_\F$ for further distinction.
For $\cA_\K$ we allow for isotropic spaces to come from
$\K^n$. 
Since isotropic spaces and
$c$-decompositions of
$\cA_\F$ naturally give isotropic spaces and $c$-decompositions of $\cA_\K$, we
have the following (see also \cite[Lemma 6]{Gelbukh}).
\begin{proposition}\label{prop:fd}
Let $\cA_\F\leq \Lambda(n, \F)$ and $\cA_\K\leq \Lambda(n, \K)$ be as above. Then
$\alpha(\cA_\F)\leq \alpha(\cA_\K)$, and $\chi(\cA_\F)\geq \chi(\cA_\K)$.
\end{proposition}

In Proposition~\ref{prop:fd}, the inqualities for $\alpha(\cdot)$ and
$\chi(\cdot)$ could be strict, as shown by Buhler, Gupta, and Harris
\cite[pp. 277]{BGH87}.
Recall that in Section~\ref{subsec:2-decomp} we defined $\alpha(\F, n,
m)=\{\alpha(\cA) : \cA\leq \Lambda(n, \F), \dim(\cA)=m\}$, and Buhler et al.
showed
that $\alpha(\F, n, m)\leq \lceil \frac{m+2n}{m+2}\rceil$, when $m>1$ and
the characteristic of $\F$ is not $2$. Furthermore the equality can be attained
for algebraically closed fields. Buhler et al. then demonstrated
examples over $\F_q$ and $\Q$, for which the inequality is strict. Consider the
example over $\Q$, which is some $\cA\leq \Lambda(n, \Q)$ of dimension $n$. They
show that $\alpha(\cA_\Q)=1$ and $\alpha(\cA_\C)=2$, which is equivalent to that
$\chi(\cA_\Q)=n$ and $\chi(\cA_\C)\leq n-1$. This gives the desired separations.
Some interesting discussions on $\R$ vs $\C$ can be found in \cite[Sec. 3]{BGH87}.

\subsection{Isotropic spaces and decompositions in group theory and 
manifold theory}\label{app:origin}

In this subsection, we explain the origins of isotropic spaces and decompositions 
for alternating matrix spaces in group theory and manifold theory. These are  
classical, so the purpose here is to provide references for interested 
readers who have not met with these before. 

\paragraph{Group theory.} 
Let $p$ be 
a prime $>2$. We 
consider $p$-groups of class $2$
and exponent $p$, that is, a group $G$ of prime power order, with the commutator
subgroup $[G, G]$ contained in the centre $Z(G)$, and every group element $g\in G$ 
satisfying
$g^p=1$.

Let $P$ be such a
group
of order $p^\ell$. We have that the
commutator subgroup $[P, P]\cong \Z_p^m$, and the commutator quotient $P/[P,
P]\cong
\Z_p^n$. The commutator bracket then gives an alternating bilinear map
$\phi:P/[P,P]\times P/[P,P]\to [P,P]$, or, after fixing bases of $[P,P]$ and
$P/[P,P]$, an alternating bilinear map $\phi:\Z_p^n\times \Z_p^n\to \Z_p^m$. On 
the other hand, given an alternating bilinear map $\phi:\F_p^n\times \F_p^n\to 
\F_p^m$, one can construct a $p$-group of class $2$ and exponent $p$, $P_\phi$, 
called the \emph{Baer group} corresponding to $\phi$, as 
follows. The group elements are $(v, u)\in \F_p^n\times \F_p^m$, and the group 
multiplication $\circ$ is by $(v_1, u_1)\circ (v_2, u_2) = (v_1+v_2, 
u_1+u_2+\frac{1}{2}\cdot \phi(v_1,v_2))$. This sets up a two-way correspondence 
between $p$-groups of class $2$ and exponent $p$ and alternating biliear maps.

Having set up the connection, let us see the group-theoretic interpretations of 
isotropic spaces and decompositions. The following is classical: an isotropic 
space of $\phi$ 
corresponds to a normal abelian subgroup of $P$ containing the commutator subgroup 
(see e.g. \cite{Alp65}). 
More recently, Lewis and Wilson proposed 
the concept of a hyperbolic pair of $P$ 
\cite{LW12}, which just
consists of
two normal abelian subgroups of $P$ which together generate $P$, and whose
intersection equals $[P,P]$. A natural generalization is then the following. A 
hyperbolic $c$-system of $P$ consists of $c$ normal
abelian subgroups $A_1, A_2$, $\dots$, $A_c$, such that $P$ is generated by $A_i$, 
and
for any $i, j\in[c]$, $i\neq j$, $A_i\cap A_j=[P,P]$. A hyperbolic $c$-system then 
naturally corresponds to an isotropic $c$-decomposition of $\phi$.



\paragraph{Manifold theory.} We then turn to the manifold 
theory side. We shall just walk through some examples of the connection, and point 
the interested reader to the survey \cite{Dim08} for more detailed information. 

Let $M$ be a compact K\"ahler manifold. Let $H^i(M; \C)$ be the $i$th cohomology 
group of $M$
with coefficients in $\C$. The cup product $\smile:H^1(M; \C)\times H^1(M; \C)\to
H^2(M; \C)$ is a skew-symmetric bilinear map. Then as an application of results by 
Castelnuovo and de Franchis, Catanese \cite{Cat91} showed that there 
exists a constant holomorphic map $f:M\to C$, where $C$ is a curve of genus $g\geq 
2$, if and only if, $\smile$ has a dimension-$g$ maximal isotropic space. Catanese 
went on to generalize this connection much further in \cite{Cat91}.

Let $M$ be a smooth closed orientable
$n$-dimensional manifold. Let $H^i(M; \Q)$ be the $i$th cohomology group of $M$
with coefficients in $\Q$. 
Gelbukh showed that an isotropic
space of $\smile$ corresponds to a geometric structure on $M$
\cite[Lemma 10, Definition 11, Theorem 13]{Gelbukh}, called an isotropic system,
which consists
of smooth closed
orientable
connected codimension-one submanifolds that are homologically non-intersecting,
homologically independent, and intersecting transversely. The isotropic index
defined in \cite{Gelbukh} for an alternating bilinear map is just the isotropic
number introduced in Definition~\ref{def:main}. In particular, our 
Theorem~\ref{thm:reduction} shows that computing this isotropic index is NP-hard. 
The relations of isotropic indices
with other basic notions in manifold cohomology, including the first Betti number
and the co-rank of the fundamental group, are also studied there.

\section{Proof of Theorem~\ref{thm:reduction}}\label{sec:reduction}

Recall that we have a graph $G=([n], E)$, and an alternating matrix
space $\cA_G$ which  is spanned by those elementary alternating matrices $A_{i,
j}$ where $\{i, j\}\in E$.

\paragraph{(1) Isotropic spaces and independent sets.} We need to show that
$G=([n], E)$
has a size-$s$ independent set if and only if $\cA_G$ has a dimension-$s$
isotropic space.

For the only if direction,
let $T=\{i_1, \dots, i_s\}$ be a size-$s$ independent
set of $G$. Let $U$ be the subspace of $\F^n$ spanned by $e_{i_1}, \dots,
e_{i_s}$; recall that $e_i$ denotes the $i$th standard basis vector of $\F^n$. It
is easy to verify that $U$ is an isotropic space of $\cA$ of dimension $s$.

For the if direction, let $U=\langle u_1, \dots, u_s\rangle$ be a dimension-$s$
isotropic space of $\cA_G$. We form an $n\times s$ matrix $U$ such that the $i$th
column of $U$ is $u_i$, that is, $U=[u_1, \dots, u_s]\in \GL(n\times s, \F)$.
Suppose $U^t=[w_1, \dots, w_n]\in \GL(s\times n, \F)$, $w_i\in \F^s$.
Since $U$ is
of rank $s$, there exist integers $i_1, \dots, i_s$, $1\leq i_1<\dots<i_s\leq n$,
such that $w_{i_1}, \dots, w_{i_s}$ are linearly independent. We now claim that
$\{i_1, \dots, i_s\}$ forms an independent set of the original graph $G=([n], E)$.
If not, suppose $\{i_j, i_{j'}\}$, $1\leq j<j'\leq s$, is in $E$.
As $U$ is an isotropic space of $\cA_G$, we have that for any
$\{k, \ell\}\in E$,
$U^tA_{k, \ell}U=\zerovec$. As
$A_{k, \ell}=e_ke_\ell^t-e_\ell e_k^t$, we have
$$U^tA_{k, \ell}U=U^t(e_ke_\ell^t-e_\ell e_k^t)U=w_kw_\ell^t-w_\ell w_k^t=0,$$
which
implies that
$w_k$
and $w_\ell$ are linearly dependent. It follows that $w_{i_j}$ and $w_{i_{j'}}$
are linearly dependent.
We then arrive at a contradiction.

\paragraph{(2) Isotropic decompositions and vertex colorings.} We need to show
that $G=([n], E)$ has a vertex $c$-coloring if and only if $\cA_G$ has an isotropic
$c$-decomposition.

For the only if direction, assume $G$ has a vertex $c$-coloring, and let
$[n]=T_1\uplus T_2\uplus \dots \uplus T_c$ be a
partition of $[n]$ into disjoint union of independent sets. Suppose $|T_j|=t_j$,
and $T_j=\{i_{j, 1}, \dots, i_{j, t_j}\}\subseteq [n]$. Let $U_j=\langle e_{i_{j,
1}}, \dots, e_{i_{j, t_j}}\rangle\leq \F^n$, so $\F^n=U_1\oplus U_2\oplus \dots
\oplus U_c$. By (1), every $U_j$ is an isotropic space. This gives an isotropic
$c$-decomposition of $\cA_G$.

For the if direction, let $\F^n=U_1\oplus U_2\oplus \dots\oplus U_c$ be an
isotropic $c$-decomposition. Let $d_i=\dim(U_i)$, and $b_i=\sum_{j=1}^id_j$, for
$i\in[c]$. Set $b_0=0$. Let
$P=[p_1, \dots, p_n]$ be an $n\times n$ invertible matrix, where $p_i\in \F^n$,
such that $p_{b_{i-1}+1}, \dots, p_{b_i}$ form a basis of $U_i$. By abuse of
notation, we also let $U_i=[p_{b_{i-1}+1}, \dots, p_{b_i}]\in \GL(n\times d_i,
\F)$, so $P=[U_1, \dots, U_c]$. Let $W_i=U_i^t=[w_{i, 1}, \dots, w_{i, n}]\in
\GL(d_i\times n, \F)$. Since $U_i$ is an isotropic space, by (1), we know that for
any $\{k, \ell\}\in E$, $w_{i, k}$ and $w_{i, \ell}$ are linearly dependent.

We then use the following simple linear algebraic result, which is a consequence
of the Laplacian expansion. For $A, B\subseteq[n]$ of the same size, we let
$C|_{A, B}$ to denote the
submatrix of $C$ with row indices from $A$ and column indices from $B$.
\begin{lemma}\label{lem:laplace}
Let $P=[U_1, \dots, U_c]\in \GL(n, \F)$, $U_i\in \GL(n\times d_i, \F)$, and suppose
$\det(P)\neq 0$. Then there exists a partition of $[n]=T_1\uplus T_2\uplus
\dots\uplus T_c$, where $|T_i|=d_i$, such that $\forall i\in[c]$, $\rk(U_i|_{T_i,
[d_i]})=d_i$.
\end{lemma}

We claim that the partition of $[n]=T_1\uplus T_2\uplus\dots\uplus T_c$ from
Lemma~\ref{lem:laplace} gives a vertex $c$-coloring of $G$. To see this, observe
that, the condition $\rk(U_i|_{T_i, [d_i]})=d_i$ is equivalent to that the vectors
$w_{i, j}$, $j\in T_i$, are linearly independent. This implies that $G$ cannot
have edges of the form $\{k, \ell\}$ where $k, \ell\in T_i$, as otherwise $w_{i,
k}$ and $w_{i, \ell}$ would be linearly dependent. Hence $T_i$ is an independent
set for
any $i\in[c]$. This completes the proof of the second part of
Theorem~\ref{thm:reduction}.

\section{An exposition of the proof of Theorem~\ref{thm:2-decomp}}\label{sec:star}

As mentioned in Section~\ref{subsec:2-decomp}, we give an exposition of the proof
of Theorem~\ref{thm:2-decomp} for $\F_q$ in \cite{BMW17}, using some ingredients
from \cite{IQ18} to handle $\R$ and $\C$. The main purpose is to give the reader a
flavor of how the so-called $*$-algebra technique is applied in this context. We
could not give all the details here, as that would be too long and unnecessary;
the interested reader may wish to go to \cite{Wil09,BW12,IQ18}, which contain
detailed proofs for using $*$-algebras to solve several closely related problems.

Recall that we are given $\cA=\langle A_1, \dots, A_m\rangle\leq \Lambda(n, \F)$,
and our goal is to find a non-trivial direct sum decomposition $\F^n=U_1\oplus
U_2$, such that $\cA|_{U_i}=\zerovec$ for $i=1, 2$. We first reduce to the
non-degenerate setting as follows. Suppose $\cA$ is degenerate, that is, there
exists $T\leq \F^n$ of dimension $n'$ such that $\cA'=\cA|_T$ is non-degenerate.
Then it is not hard to
verify that $\cA$
admits an isotropic $2$-decomposition if and only if $\cA'$ admits an isotropic
$2$-decomposition.

In the following we assume that $\cA$ is non-degenerate. Let $\bA=(A_1, \dots,
A_m)\in \Lambda(n, \F)^m$. The adjoint algebra of $\bA$ is defined as
$\Adj(\bA):=\{D\in \M(n, \F) : \exists B\in \M(n, \F), \forall i\in[m],
B^tA_i=A_iD\}$. Since $\bA$ is non-degenerate, if such $B$ exists, then it is
unique. Then a natural involution (an anti-automorphism of order at most $2$) on
$\Adj(\bA)$ is to send $D\in \Adj(\bA)$ to this unique $B$ satisfying
$B^tA_i=A_iD$ for any $i\in[m]$, also denoted as $D^*$. The adjoint algbera is the
key device for the algorithm.

We now translate the isotropic $2$-decomposition problem for $\cA$, and therefore
$\bA$, to a problem about $\Adj(\bA)$. Any direct sum decomposition
$\F^n=U_1\oplus U_2$ can be encoded as a projection matrix $P\in \M(n, \F)$, that
is, $P^2=P$, $\im(P)=U_1$, $\ker(P)=U_2$. The key observation in \cite{BMW17} is
that, $P$ corresponds to an isotropic $2$-decomposition if and only if $P\in
\Adj(\bA)$ and $P^*=I-P$. This means that we need to search for an idempotent $P$
in $\Adj(\bA)$ satisfying $P^*=I-P$. Following \cite{BMW17}, we call such an
idempotent a hyperbolic idempotent.

To do that, we utilize the $*$-algebra structure of $\Adj(\bA)$. For this, we
recall in a nutshell the structure of $*$-algebras. Let $\fA$ be a $*$-algebra.
The Jacobson
radical of $\fA$, denoted by $\rad(\fA)$, is the largest nilpotent
ideal
of $\fA$, and it is invariant under $*$. The factor algebra $\fA/\rad(\fA)$ is
semi-simple, namely it is a direct sum of simple algebras. Let
$\fA/\rad(\fA)\cong S_1\oplus \dots \oplus S_k$, where each $S_i$ is simple. The
$*$ either switches between $S_i$ and $S_j$, $i\neq j$, or preserves $S_i$. Both
cases are referred
to as $*$-simple. The
former case is called the exchange type. In the
latter case, $S_i$ is a simple algebra with an involution.  Over any field,
Wedderburn's theory gives a characterization of such simple algebras (see e.g.
\cite[Chap. 5]{AB95}). Based on this, involutions on simple algebras are also
classified \cite[Chap. X.4]{Alb39}, and explicit lists for $\F_q$, $\R$, and $\C$
can be found in \cite{IQ18}.

Given this structure, the idea is to reduce the search for a hyperbolic idempotent
from general $*$-algebras to simple $*$-algebras. Clearly, if $\fA$ contains a
hyperbolic idempotent, then $\fA/\rad(\fA)$, and each $*$-simple summand of
$\fA/\rad(\fA)$, all contain hyperbolic idempotents. On the other hand, suppose
each $*$-simple summand of $\fA/\rad(\fA)$ contains a hyperbolic idempotent. Then
the sum of these idempotents is a hyperbolic idempotent for $\fA/\rad(\fA)$. From
here, to obtain a hyperbolic idempotent for $\fA$, we can use the classical
idempotent lifting technique (see e.g. \cite[Lemma 5.10]{Wil09b}).

Therefore, it remains to handle the $*$-simple case. Let $\K$ denote some
appropriate division algebra containing $\F$. The reader may as well think of
$\K$ as an extension field, as for $\F_q$, $\R$, and $\C$, the only ``non-field''
case is the quaternion algebra over $\R$. The exchange type is easy to
handle: it is $*$-isomorphic to $\M(\ell, \K)\oplus \M(\ell, \K)^{op}$ with $(A,
B)^*=(B, A)$. So one hyperbolic idempotent can be $(I, \zerovec)$. The simple case
is more interesting. It is isomorphic to $\M(\ell, \K)$ with the involution defined
by some non-degenerate classical form $F$; this includes alternating, symmetric,
and
Hermitian ones.\footnote{While in principle this is correct, depending on the
field, some type may not exist. For details see \cite{IQ18}.} Then for $A\in
\M(\ell,
\K)$, $A^*=F^{-1}A^\dagger F$, where
$\dagger$ denotes either transpose (for alternating and symmetric) or conjugate
transpose (for Hermitian). The problem is then to find a hyperbolic idempotent, or
equivalently, an isotropic $2$-decomposition, for
this form $F$. But now this is a single form, so one can bring it to say a
canonical form, and examine case by case.

For example, over $\C$, there are two types, symmetric and alternating. (The
Hermitian type does not appear because $*$ is required to preserve $\C$.) A
non-degenerate alternating form can always be transformed to $\begin{bmatrix}
\zerovec
& I \\ -I & \zerovec\end{bmatrix}$, so isotropic $2$-decomposable. A
non-degenerate symmetric form can be transformed to the identity matrix $I$. When
$I$ is of odd size $\ell$, then it does not admit an isotropic $2$-decomposition.
Because if so, then $I$ is isometric to $J=\begin{bmatrix} \zerovec & A \\ A^t &
\zerovec\end{bmatrix}$, where $A$ is of size $i\times (\ell-i)$. But then $\rk(J)$
is
either $2i$ or $2(\ell-i)$, an even number, so $J$ is degenerate, a contradiction.
When $I$ is of even size $\ell$, then it has an isotropic $2$-decomposition. This
is
because we have $\begin{bmatrix} 1 & i \\ i & 1 \end{bmatrix}\begin{bmatrix} 1 & 0
\\ 0 & 1 \end{bmatrix}\begin{bmatrix} 1 & i \\ i & 1 \end{bmatrix}=\begin{bmatrix}
0 & 2i \\ 2i & 0 \end{bmatrix}$. We can then use this to bring $I$ to
$\begin{bmatrix} \zerovec & 2i I' \\ 2i I' & \zerovec\end{bmatrix}$, where $I'$ is
the $\ell/2\times \ell/2$ identity matrix. Similar reasonings can be carried over
$\R$
and $\F_q$.

To make the above procedure constructive, we need to compute the algebra structure
efficiently. This can be done, over $\F_q$ with randomness \cite{Ron90} and over
$\C$ deterministically \cite{Ebe91,Ron94}. We also need to compute the $*$-algebra
structure, e.g. the forms associated with a simple $*$-algebra, by
\cite{Wil09,BW12,IQ18}. Finally, we need to compute the canonical forms of various
forms by \cite{Wil13Gram,Woe05}. 

For the algorithm analysis, we work in the exact model, as going to extension
fields is already unavoidable in computing the algebra structure. Over $\F_q$, one
can always recover, from the solutions in the simple cases, an explicit hyperbolic
projection matrix $P$ over the original field, and then we can obtain the bases of
the two subspaces in an isotropic $2$-decomposition by computing the image and
kernel of $P$.  Over $\R$ and $\C$, one can
represent this projection matrix as a product of matrices over different extension
fields \cite[Sec. 3.5]{IQ18}.

This
concludes an exposition of the algorithm.

\section{Proof of Theorem~\ref{thm:maximum_bipartite}}\label{sec:maximum_bipartite}

Let us first recall the alternative definition of non-commutative rank for a
slightly
more general situation. Given $\cB\leq \M(s\times t,
\F)$, an isotropic pair is a pair of vector spaces $(U, V)$, $U\leq \F^s$, $V\leq
\F^t$, such that for any $u\in U, v\in V$, and any $B\in \cB$, we have $u^tBv=0$.
The non-commutative rank of $\cB$ is then defined as $\ncrk(\cB):=(s+t)-\max\{c+e
: c=\dim(U),
e=\dim(V), (U, V)\text{ is an isotropic pair of }\cB\}$. 
Note that the recent works \cite{GGOW16,IQS17} mostly deal with the setting that 
$s=t$.
Suppose $\ncrk(\cB)=r$. By equivalence
transformations, we can assume that every
$\cB$ is of the form $\begin{bmatrix} B_1 & B_2 \\ \zerovec & B_3\end{bmatrix}$,
where $B_2$ is of size $a\times b$ such that $a+b=r$.

We review the setting for Theorem~\ref{thm:maximum_bipartite}. Let $\cA\leq
\Lambda(n, \F)$ be a bipartite alternating matrix space. By isometric
transformations, we
can assume that every $A\in \cA$ is of the form $\begin{bmatrix} \zerovec & B \\
-B^t
& \zerovec\end{bmatrix}$ where $B\in \M(s\times t, \F)$, $s+t=n$. All such $B$
form
a matrix space $\cB\leq \M(s\times t, \F)$. We call such $\cB$ a matrix space
induced from the bipartite structure of $\cA$.

Before proving Theorem~\ref{thm:maximum_bipartite}, let us examine some examples
of isotropic spaces of
$\cA$.
\begin{enumerate}
\item First note that
$\alpha(\cA)\geq \max\{s, t\}$.
\item Second, suppose $\ncrk(\cB)=r$, so there exists
$P\in \GL(s, \F)$ and $Q\in \GL(t, \F)$, such that every matrix in $P\cB Q$ is of
the form $\begin{bmatrix} B_1 & B_2 \\ \zerovec & B_3\end{bmatrix}$,
where $B_2$ is of size $a\times b$ such that $a+b=r$. Let $R=\begin{bmatrix}P^t &
\zerovec \\ \zerovec & Q \end{bmatrix}$. Then we have
\begin{eqnarray*}
R^tAR & = & \begin{bmatrix} P & \zerovec \\ \zerovec & Q^t \end{bmatrix}
\begin{bmatrix} \zerovec & B \\ -B^t & \zerovec\end{bmatrix} \begin{bmatrix} P^t &
\zerovec \\ \zerovec & Q \end{bmatrix} \\
& = & \begin{bmatrix} \zerovec & PBQ \\ -(PBQ)^t & \zerovec\end{bmatrix} \\
& = & \begin{bmatrix} \zerovec & \zerovec & B_1 & B_2 \\ \zerovec & \zerovec &
\zerovec & B_3 \\ -B_1^t & \zerovec & \zerovec & \zerovec \\ -B_2^t & -B_3^t &
\zerovec & \zerovec\end{bmatrix},
\end{eqnarray*}
from which we get an isotropic space, consisting of the zero blocks in the middle
part, of dimension $n-r$. (Note that the zero matrix at the $(2, 3)$ position is
of size $(s-a)\times (t-b)$, and the isotropic space corresponding to the zero
blocks in the middle part is of size $(s-a)+(t-b)=(s+t)-(a+b)=n-r$.)
\item The third example we now describe represents a difference from the
graph-theoretic
setting. Suppose $s=t$, and every $B\in \M(s, \F)$ is symmetric. So $n=2s$. Then
let $U=\{u\in \F^n=\F^{2s} : u=(u_1, \dots, u_s, u_1, \dots, u_s)^t\in \F^n\}\leq
\F^n$.
That is, $U$ consists of those vectors whose $i$th component equals the $(i+s)$th
component, for $i\in [s]$, and $\dim(U)=s$. We claim that $U$ is an isotropic
space of $\cA$. To
see this, for a given $u\in U$, we can write it as $\begin{bmatrix} v \\
v\end{bmatrix}$ where $v\in \F^s$. So for $A\in \cA$ such that
$A=\begin{bmatrix}\zerovec & B \\ -B^t & \zerovec \end{bmatrix}$, we have
$
 \begin{bmatrix}v^t & v^t\end{bmatrix}\begin{bmatrix}\zerovec & B \\ -B^t &
\zerovec \end{bmatrix} \begin{bmatrix} v \\ v\end{bmatrix}
= -v^tB^tv+v^tBv=-v^tBv+v^tBv=0$.
\end{enumerate}
The proof of Theorem~\ref{thm:maximum_bipartite} basically suggests that isotropic
spaces of
the third type are not going to matter for the comparison with $\alpha(\cA)$.
We now go into the proof.
\begin{proof}[Proof of Theorem~\ref{thm:maximum_bipartite}] Let $r=\ncrk(\cB)$ and
$d=\alpha(\cA)$.

We first show that $\alpha(\cA)\geq n-\ncrk(\cB)$. Note that $r\leq
\min\{s, t\}$, because we have the trivial isotropic pairs $(\F^s, \zerovec)$ and
$(\zerovec, \F^t)$. If $r=\min\{s, t\}$, note that $n-\min\{s,
t\}=\max\{n-s, n-t\}=\max\{s, t\}$, and we do have isotropic spaces of dimensions
$s$
and $t$, respectively, by (1) from above. If $r<\min\{s, t\}$, then by (2) from
above, there is an isotropic space of dimensions $n-r$. This shows that
$\alpha(\cA)\geq n-\ncrk(\cB)$.

We then show that $\alpha(\cA)\leq n-\ncrk(\cB)$, or equivalently, $\ncrk(\cB)\leq
n-\alpha(\cA)$. Again, if $\alpha(\cA)=\max\{s, t\}$, then $\ncrk(\cB)\leq
n-\max\{s, t\}=\min\{s, t\}$, which trivially holds. So we assume that
$\alpha(\cA)>\max\{s, t\}$.
Let $U$ be an isotropic space of
dimension $d=\alpha(\cA)$, and take an $n\times d$ matrix whose columns form a
basis of $U$, which, by abuse of notation, is also denoted by $U$. Let
$U=\begin{bmatrix} v_1' & v_2' & \dots & v_d' \\ w_1' & w_2' & \dots &
w_d'\end{bmatrix}$, where $v_i'\in \F^s$, and $w_i'\in \F^t$. Let
$V'=\begin{bmatrix} v_1'&v_2'&\dots&v_d'\end{bmatrix}$, and $W'=\begin{bmatrix}
w_1'&w_2'&\dots&w_d'\end{bmatrix}$. By doing a linear
combination over the columns, we can assume that $U$ is of the form
$\begin{bmatrix}
V & \zerovec \\
W'' & W
\end{bmatrix}$ where $V\in \M(s\times c, \F)$, $W''\in \M(t\times c, \F)$, and
$W\in \M(t\times e, \F)$, such that $c+e=d$, $\rk(V)=c$, and $\rk(W)=e$. By abuse
of notation, let $V$ be the subspace of $\F^s$ spanned by columns in $V$, and $W$
the subspace of $\F^t$ spanned by columns in $W$.
\begin{claim}
Let $V$ and $W$ be as above. Then $(V, W)$ form an
isotropic pair for $\cB$.
\end{claim}
\begin{proof}
From above we have that
$U=\begin{bmatrix}
V & \zerovec \\
W'' & W
\end{bmatrix}$, and suppose $U=
\begin{bmatrix} v_1 & v_2 & \dots & v_d \\ w_1 & w_2 & \dots &
w_d\end{bmatrix}$. This means that $v_i=\zerovec$ for $c<i\leq d$.
Since $U$ is an isotropic space of $\cA$, we have, for
any $i, j\in [d]$, and
$A=\begin{bmatrix} \zerovec & B \\-B^t & \zerovec\end{bmatrix}\in \cA$,
$\begin{bmatrix} v_i^t & w_i^t\end{bmatrix}\begin{bmatrix} \zerovec & B \\-B^t &
\zerovec\end{bmatrix}
\begin{bmatrix} v_j \\
w_j\end{bmatrix}=-w_i^tB^tv_j+v_i^tBw_j=-v_j^tBw_i+v_i^tBw_j=0$, so
$v_i^tBw_j=v_j^tBw_i$. In particular, for any column vector $v\in V$ and any
column vector
$w\in W$, we have $v^tBw=\zerovec^tBw''=0$ for some column vector $w''\in W''$.
This
concludes the proof.
\end{proof}
We then get that $\ncrk(\cB)\leq (s+t)-(c+e)=n-d=n-\alpha(\cA)$, and the proof is
concluded.
\end{proof}

We now set out to prove Corollary~\ref{cor:maximum_bipartite}. For this we need
one more ingredient.
In the literature \cite{GGOW16,IQS17},
the computation of the non-commutative ranks only deals with the case when $s=t$.
To use that for general $\M(s\times t, \F)$ we need a little twist.
\begin{proposition}\label{prop:ncrk}
Over any field $\F$, computing the non-commutative rank of $\cB\leq \M(s\times t,
\F)$ can be done in deterministic polynomial time.
\end{proposition}
\begin{proof}
If $s=t$, this follows from \cite{IQS17}. Without loss of generality, let us
assume then $s<t$. We shall construct some $\cC\leq \M(t, \F)$ as follows. First,
for any $B\in \cB$, let $B'=\begin{bmatrix}
\zerovec\\ B\end{bmatrix}$, where $\zerovec$ is of size $(t-s)\times t$, so $B'\in
\M(t, \F)$. Second, recall that $E_{i,j}$ is the elementary matrix with the
$(i,j)$th entry being $1$, and the rest entries being $0$. Then $\cC$ is the
matrix space spanned by all $B'$ and $E_{i,j}$ with $1\leq i\leq t-s$ and $1\leq
j\leq t$.

We claim that
$\ncrk(\cB)+(t-s)=\ncrk(\cC)$. To see that $\ncrk(\cB)+(t-s)\geq \ncrk(\cC)$, let
$(U, V)$ be an isotropic pair of $\cB$, where $U\leq \F^s$, $V\leq \F^t$, such
that $\ncrk(\cB)=s+t-\dim(U)-\dim(V)$. Let $U'\leq \F^t$ be the image of $U$ under
the embedding $\F^s$ to $\F^t$ by sending $e_i$ to $e_{i+t-s}$. Clearly, $(U', V)$
is an isotropic pair for $\cC$, so
$\ncrk(\cC)\leq 2t-\dim(U')-\dim(V)=(t-s)+s+t-\dim(U)-\dim(V)=(t-s)+\ncrk(\cB)$.

To show that $\ncrk(\cB)+(t-s)\leq \ncrk(\cC)$, let $(U, V)$ be an isotropic pair
of $\cC$. If $\ncrk(\cC)=t$, then the equality is trivial. So in the following we
assume $\ncrk(\cC)<t$. We claim that if $V\neq
\zerovec$, then $U$ is a subspace of $\langle
e_{t-s+1}, \dots,
e_t\rangle$. Suppose not, then $U$ contains a vector $u=[u_1, \dots, u_t]^t$
with
some $u_i\neq 0$ for $1\leq i\leq t-s$. Because $E_{i,j}$ is present in $\cC$ for
$1\leq j\leq t$, for $v\in \F^n$ to satisfy that $u^tE_{i,j}v=0$ for any $1\leq
j\leq t$, $v$ has to be $\zerovec$. This implies that $V$ has to be $\zerovec$.
Therefore all isotropic pairs of $\cC$, except a trivial one $(\F^t, \zerovec)$,
are also isotropic pairs of $\cB$. Therefore, if $\ncrk(\cC)<t$, and $(U, V)$ is
an
isotropic pair for $\ncrk(\cC)$ such that $\ncrk(\cC)=2t-\dim(U)-\dim(V)<t$, we
have
$V\neq \zerovec$, so $U\leq \langle
e_{t-s+1}, \dots,
e_t\rangle$. Let $U'$ be the image of $U$ under the projection from $\F^t$ to
$\F^s$ by
sending
$(v_1, \dots, v_t)^t$ to $(v_{t-s+1}, \dots, v_t)^t$. We see then
$\dim(U')=\dim(U)$, and $(U', V)$ is an isotropic space for $\cB$. From
this we can conclude the proof.
\end{proof}

We are now ready to prove Corollary~\ref{cor:maximum_bipartite}.
\begin{proof}[Proof of Corollary~\ref{cor:maximum_bipartite}]
Given $\cA\leq \Lambda(n, \F_q)$, $q$ odd,  first put it into the explicit
bipartite form
using Theorem~\ref{thm:2-decomp}, which also produces the bases of the two
subspaces in an isotropic $2$-decomposition.
Then for a matrix space $\cB\leq \M(s\times t, \F)$ induced from the
bipartite structure, compute its non-commutative rank $r=\ncrk(\cB)$ using
Proposition~\ref{prop:ncrk}. The isotropic number of $\cA$ is then $n-r$, by
Theorem~\ref{thm:maximum_bipartite}.
\end{proof}

\begin{remark}
To obtain analogues of Corollary~\ref{cor:maximum_bipartite} over $\R$ and $\C$,
the bottleneck is that over $\R$ and $\C$, Theorem~\ref{thm:2-decomp} only outputs
the projection matrix as the product of a sequence of matrices over different
extension fields. This prevents us from working with the bases of the subspaces in
an isotropic $2$-decomposition directly. Of course, if we are content with
approximating those algebraic numbers up to certain precision, we can use the
representation of the
projection as a product of matrices over different extension fields to get such,
and then work with that.
\end{remark}

\section{Proof of Theorem~\ref{thm:iso_dim_2}}\label{sec:iso_dim_2}

Let us first work with general $\F$, and then restrict to our target field $\Q$ at
some point.
Recall that the problem is to decide
whether an alternating matrix space $\cA\leq \Lambda(n, \F)$ has an isotropic
space of dimension $2$. We first make the following easy observation.
\begin{observation}\label{obs:iso_dim_2}
Let $\cA=\langle A_1, \dots, A_m\rangle\leq \Lambda(n, \F)$. Then $\cA$ has an
isotropic space of dimension $2$,
if and only if there exist linearly independent $v, w\in \F^n$ such that for any
$A\in \cA$, $v^tAw=0$, which is further equivalent to that for any $i\in[m]$,
$v^tA_iw=0$.
\end{observation}

We now need to prove some auxiliary results. Let $\cB=\langle B_1, \dots,
B_m\rangle\leq \M(n, \F)$, and let $\vB=(B_1, \dots,
B_m)\in \M(n, \F)^m$. Here
is a natural problem about matrix spaces.
\begin{problem}\label{prob:exist_sing}
The existential singularity problem for matrix spaces, or the linear
$\exists$-singularity
problem, asks the following: given $\cB\in \M(n, \F)$,
decide
whether there exists a singular (e.g. non-full-rank) \emph{non-zero} matrix in
$\cB$.
\end{problem}

The linear $\exists$-singularity problem turns out to be quite interesting. We
discuss this problem in detail in Section~\ref{subsec:exist_sing}. For the sake of
proving Theorem~\ref{thm:iso_dim_2}, we need the
following result, whose proof can be found there.
\begin{lemma}\label{lem:factoring}
Over $\Q$, assuming the generalized Riemann hypothesis, there is a randomized
polynomial-time
reduction from deciding quadratic residuosity modulo squarefree composite
numbers
to the linear $\exists$-singularity problem.
\end{lemma}

We then show that the linear $\exists$-singularity problem reduces to deciding
whether an
alternating matrix space has an isotropic space of dimension $2$. This reduction
works over any field.

For
this purpose, it will be convenient to define an intermediate problem,
which may be viewed as just a reformulation of Problem~\ref{prob:exist_sing}.

Recall that $\cB=\langle B_1, \dots, B_m\rangle\leq\M(n, \F)$, and $\vB=(B_1,
\dots, B_m)\in \M(n, \F)^m$. Think of $\vB$ as a $3$-tensor $T_\vB$ of size
$n\times n\times m$,
such that $T_\vB(i, j, k)=B_k(i, j)$. That is, $B_k$'s are the slices according to
the
third index (lateral slices). We will also be interested in the matrices obtained
according to the first index (horizontal slices) and the second index (vertical
slices). Specifically, define $\vB^v$ be the $n$-tuple of $n\times m$ matrices
that are vertical slices of $T_\vB$. That is, $\vB'=(B'_1, \dots, B'_n)\in
\M(n\times m, \F)^n$, where $B'_j=[B_1e_j, \dots, B_me_j]$, or in other words,
$B'_j(i, k)=T_\vB(i,j,k)=B_k(i,j)$.
Similarly we can define the matrix tuple
consisting of the horizontal slices of $T_\vB$.

We now consider the matrix space $\vB'\in \M(n\times m, \F)^n$. For $v=(v_1,
\dots, v_m)^t\in \F^m$,
its right degree in $\vB'$ is defined to be the rank of
$[B'_1v, \dots, B'_mv]=v_1B_1+\dots+v_mB_m$. Therefore, every non-zero $v\in \F^n$
has right degree $n$ in $\vB'$, if and only if every matrix in $\cB$ is of rank
$n$.
Lemma~\ref{lem:factoring} then
immediately implies the following.
\begin{corollary}\label{cor:iso_dim_2}
Over $\Q$, assuming the generalized Riemann hypothesis, there is a randomized
polynomial-time
reduction from deciding quadratic residuosity modulo squarefree composite
numbers
to deciding whether there exists a non-zero $v\in \F^m$ of right degree $<n$
w.r.t. a matrix tuple $\vB'\in
\M(n\times
m, \F)$.
\end{corollary}

Given $\vB'=(B_1', \dots, B_n')\in \M(n\times m, \F)$, we construct a tuple of
alternating matrices of size $(n+m)\times (n+m)$, as
follows. For $i\in[n]$, let $A_i=\begin{bmatrix}
\zerovec & B_i' \\
-B_i'^t & \zerovec
\end{bmatrix}$. For $1\leq i<j\leq n$, let $C_{i,
j}=e_ie_j^t-e_je_i^t$.
For $1\leq k<\ell\leq m$, let $D_{k,
\ell}=e_{n+k}e_{n+\ell}^t-e_{n+\ell}e_{n+k}^t$.
Note that $C_{i,j}$ and $D_{k, \ell}$ are elementary alternating matrices.
Let $\bA=(A_1,
\dots, A_n, C_{1, 2}, \dots, C_{n-1, n}, D_{1, 2}, \dots, D_{m-1, m})$. We now
claim the following.
\begin{claim}\label{claim:iso_dim_2}
Let $\vB'\in \M(n\times m, \F)^n$ and
$\bA\in \Lambda(n+m, \F)^{n+\binom{n}{2}+\binom{m}{2}}$ be as above. Let
$\cA=\langle \bA\rangle \leq \Lambda(n+m, \F)$. Then there exists a non-zero
$v'\in \F^m$ of
right degree $<n$ in $\vB'$ if and only if $\cA$ has an isotropic space of
dimension $2$.
\end{claim}
\begin{proof}
By Observation~\ref{obs:iso_dim_2}, to decide whether $\cA$ has an isotropic space
of dimension $\geq 2$, we only need to test whether there exist
linearly independent $u, v\in \F^{n+m}$, such that for any $E=A_i$, or $C_{i, j}$,
or $D_{k, \ell}$,
$u^tEv=0$.

The only if direction is easy to verify. Suppose $v'\in \F^m$ is of right degree
$n-1$ w.r.t. $\vB'$, namely $[B'_1v', \dots, B'_mv']$ is of rank $<n$. Then take
any non-zero $u'\in \F^n$ in the left kernel of $[B'_1v', \dots, B'_mv']$, and we
have that for $i\in[m]$, $u'^tB'_iv'=0$. Now construct $u=\begin{bmatrix} u' \\
\zerovec\end{bmatrix}\in \F^{n+m}$, and $v=\begin{bmatrix} \zerovec \\
v'\end{bmatrix}\in \F^{n+m}$.
For $i\in[n+m]$, let the $i$th
component of $u$ (resp. $v$) be $u_i$ (resp. $v_i$). Then for $i\in \{n+1, \dots,
n+m\}$, $u_i=0$. For $i\in [n]$, $v_i=0$. Clearly, $u$ and $v$ are
linearly independent. Furthermore, it is easy to verify that (1)
$u^tA_iv=u'^tB'_iv'=0$, (2) $u^tC_{i,j}v=u_iv_j-u_jv_i=u_i\cdot 0-u_j\cdot 0=0$,
as $i, j\in[n]$,
and similarly (3) $u^tD_{k, \ell}v=0$. Then $u$ and $v$ spans a dimension-$2$
isotropic space of $\cA$.

For the if direction, suppose $u$ and $v$ are linearly independent vectors in
$\F^{n+m}$, and satisfy that for any $E=A_i$, or
$C_{i, j}$,
or $D_{k, \ell}$,
$u^tEv=0$. Write $u=\begin{bmatrix} u_1 \\ u_2\end{bmatrix}$, where $u_1\in \F^n$
and $u_2\in \F^m$.
Similarly write $v=\begin{bmatrix} v_1\\ v_2\end{bmatrix}$, where $v_1\in \F^n$
and $v_2\in \F^m$. By
$u^tC_{i,j}v=0$, we have that $u_1$ and $v_1$ are linearly dependent. By
$u^tD_{k,\ell}v=0$, we have that $u_2$ and $v_2$ are linearly dependent. We first
observe that it cannot be the case that $u_1=v_1=\zerovec$, nor
$u_2=v_2=\zerovec$.
As otherwise, say if $u_1=v_1=\zerovec$, then because $u_2$ and $v_2$ are linearly
dependent, we have $u$ and $v$ are linearly dependent, a contraction. Therefore,
without loss of generality,
we assume that $u_1\neq \zerovec$, so $v_1=\alpha_1 u_1$ for some $\alpha_1\in
\F$. We then have two cases. In the first case, if $u_2=\zerovec$, then $v_2\neq
\zerovec$, and $v'=v-\alpha_1 u$ and $u$ are linearly independent. In the second
case, if $u_2\neq \zerovec$, then $v_2\neq \alpha_1u_2$, as otherwise $u$ and $v$
would be linearly dependent. Then again, letting $v'=v-\alpha_1 u$, we have $u$
and $v'$ are linearly independent.
Clearly, $u$ and $v'$ still satisfy that for any $E=A_i$, or $C_{i, j}$, or $D_{k,
\ell}$, $u^tEv'=0$. Write $v'$ as $\begin{bmatrix} v'_1 \\ v'_2\end{bmatrix}$
where $v'_1\in \F^n$, and
$v'_2\in \F^m$, so
$v'_1=\zerovec$, and $v'_2\neq \zerovec$. We then have $u_2=\alpha_2v'_2$. Letting
$u'=u-\alpha_2v'$, we
have $u'$ and $v'$ are linearly independent, and for any $E=A_i$, or $C_{i, j}$,
or $D_{k,
\ell}$, $u'^tEv'=0$. Write $u'$ as $\begin{bmatrix}u'_1 \\ u'_2\end{bmatrix}$
where $u'_1\in \F^n$, and
$u'_2\in \F^m$, so
$u'_1\neq \zerovec$, and $u'_2= \zerovec$. It is then straightforward to verify
that the condition $u'^tA_iv'=0$ is equivalent to $u_1'^tB_i'v'_2\neq 0$. Recall
that neither $u'_1$ nor $v'_2$ is the zero vector; this just translates to say
that $v'_2$ is of right degree $<n$ w.r.t. $\vB'$.
\end{proof}

Theorem~\ref{thm:iso_dim_2} follows by combining Claim~\ref{claim:iso_dim_2} and
Corollary~\ref{cor:iso_dim_2}.

\subsection{The existential singularity problem for matrix
spaces}\label{subsec:exist_sing}

In this subsection, we discuss on Problem~\ref{prob:exist_sing}, which we believe
is a very interesting problem in its own right. We therefore examine this problem
over various fields, and prove Lemma~\ref{lem:factoring} over $\Q$.

The affine version of Problem~\ref{prob:exist_sing} has been studied in
\cite{BFS99}. More specifically, the $\exists$-singularity problem for affine
matrix spaces asks to decide whether an affine
matrix space contains a non-full-rank matrix (not necessarily nonzero). In
\cite{BFS99}, this problem was called the singularity problem. This may cause some
confusion, because in \cite{GGOW16,IQS17} the singularity problem for matrix
spaces is to decide whether all matrices in a matrix space are singular.
For clarification, we then call the problem in \cite{BFS99} the affine
$\exists$-singularity problem, and Problem~\ref{prob:exist_sing} the linear
$\exists$-singularity problem.

In \cite{BFS99}, it was shown that the affine $\exists$-singularity problem is
NP-hard over $\F_q$, $\Q$, or $\R$. We first note that the linear
$\exists$-singularity problem reduce to that for affine matrix spaces, but
the inverse direction is not clear.\footnote{To reduce the matrix space case to
the affine case is easy: if $\cB=\langle B_1, \dots, B_m\rangle$, then form $m$
affine spaces $B_i+\langle B_1, \dots, B_{i-1}, B_{i+1}, \dots, B_m\rangle$.}
Furthermore, the proof strategy of \cite{BFS99}
cannot be adapted directly to tackle Problem~\ref{prob:exist_sing}, because of the
introduction of field constants in the reduction. In particular, the use
of field constants in the transformation from algebraic branching programs to
symbolic determinants by Valiant \cite{Val79} seems particularly crucial.
Indeed, as we will see below, the $\exists$-singular problems for matrix spaces
and for affine matrix spaces demonstrate quite different behaviors.

Matrix spaces in which every nonzero matrix is of full-rank has been studied in
mathematics for a long time. More broadly, if $\cB\leq \M(n, \F)$ satisfies that
every nonzero
matrix in $\cB$ is of a fixed rank $r$, we say that $\cB$ satisfies the fixed rank
$r$ condition. Such matrix spaces are of interests in algebraic geometry (mostly
when over algebraically closed fields), differential topology (mostly when over
$\R$), number
theory (mostly when over $\Q$), and
algebra (mostly when over finite fields). It turns out that several results from
these
different
branches of mathematics will be useful for our algorithmic purposes as well.

To start with, the following quantity has been studied extensively in the
literature. Let $\rho(n, r, \F)$ be the
maximum dimension over those $\cB\leq \M(n, \F)$ satisfying the fixed rank $r$
condition. Also let $\tau(n, r, \F)$ be the maximum dimension over those affine
matrix spaces $\cC\subseteq \M(n, \F)$ satisfying the fixed rank $r$ condition.
Also let $\rho(n, \F):=\rho(n, n, \F)$, and $\tau(n, \F):=\tau(n, n, \F)$. As
pointed out in \cite{dSP13}, $\rho(n, \F)\leq n$, and $\tau(n, \F)=\binom{n}{2}$.
Two remarks are due here. First, $\rho(n, \F)$ can be much smaller than $n$ for
certain fields; see below. Second, the $\binom{n}{2}$ bound for $\tau(n, \F)$ can
be easily achieved at $I_n+\mathcal{U}$ where $\mathcal{U}$ is the linear space of
strictly upper triangular matrices. This distinction already suggests that the
difference between the linear and the affine cases can be significant.

In this following, we discuss on $\C$, $\R$, and $\Q$, comparing the linear and
affine settings, and presenting some algorithms for the linear case, including a
proof of Lemma~\ref{lem:factoring}. We refer the interested reader to \cite{She11}
for the
finite field case.

\paragraph{Over $\C$.} The affine $\exists$-singularity
problem over $\C$ is only known to be in RP \cite{BFS99}.

We then turn to the linear $\exists$-singularity problem over $\C$. Sylvester
showed that $\rho(n,
\C)\leq 3$ \cite{Syl86}, and Westwick generalized that to $\rho(n, r, \C)\leq
2n-2r+1$ \cite{Wes87}. Some
subsequent
developments include \cite{IL99,BFM13}.

Sylvester's
result immediately translates to a deterministic efficient algorithm for the
linear $\exists$-singularity problem over $\C$: if the input matrix space $\cB\leq
\M(n,
\C)$ is of dimension $\geq 4$, then return ``exists.'' Otherwise, $\cB=\langle
B_1, B_2, B_3\rangle$. We can then write out $\det(x_1B_1+x_2B_2+x_3B_3)$, and use 
algorithms from computational algebraic geometry \cite{CLO07} to decide whether 
this polynomial only has the trivial zero. 

\paragraph{Over $\R$.} The affine $\exists$-singularity problem over $\R$ is
NP-hard
\cite{BFS99}.

We then turn to the linear $\exists$-singularity problem over $\R$. Based on the
Radon-Hurwitz construction and Adams' vector
field theorem \cite{Ada62},
$\rho(n, \R)$ is equal to the so-called Hurwitz-Radon function (see
\cite{ALP65a}).
For
$n\in \N$,
write $n$ in the form of $2^{4a+b}\cdot(2c+1)$ where $b\in\{0, 1, 2, 3\}$, and the
Hurwitz-Radon function is $HR(n)=8a+2^b$. The significance of this result for our
algorithmic purpose is that $HR(n)\leq 2(\log n+4)$. Some subsequent developments
include \cite{Mes90,LY93,Cau07}.

Therefore, the linear
$\exists$-singularity problem over $\R$ admits the following quasipolynomial-time
algorithm. If the input matrix space $\cB\leq \M(n, \C)$ is of dimension $\geq
HR(n)$, then return ``exists.'' Otherwise, $\cB=\langle B_1, \dots, B_m\rangle$
where $m\leq HR(n)\leq 2(\log n+4)$. We form $m$ affine spaces,
$\cC_i:=B_i+\langle B_1, \dots, B_{i-1}, B_{i+1}, \dots, B_m\rangle$, for every
$i\in[m]$. The question then becomes whether any of the $\cC_i$'s contains a
singular matrix. This can be done by computing
$f_i:=\det(B_1x_1+\dots+B_{i-1}x_{i-1}+B_i+B_{i+1}x_{i+1}+\dots +B_mx_m)$
explicitly. Since the polynomial $f_i$, involving $O(\log n)$ variables, is of
degree $n$, $f_i$ has $n^{O(\log n)}$ monomials, and we can fully write out $f_i$
in time polynomial in $n^{O(\log n)}$.\footnote{There are several ways of doing
this. One approach is to transform the determinant expression into an arithmetic
circuit, and then compute along this circuit to get the final polynomial.} After
that, we can use the existential theory of reals \cite{Can88,Ren92} to determine
whether $f_i$ has a non-trivial
zero in time $n^{O(\log n)}$. Return ``exists'' if and only if one of these
$f_i$'s is solvable. This concludes the proof.

\paragraph{Over $\Q$.} The affine $\exists$-singularity problem over $\Q$ is
NP-hard
\cite{BFS99}.

We then turn to the linear $\exists$-singularity problem over $\Q$. To start with,
observe
that $\rho(n, \Q)\geq n$. This is because we can take a degree-$n$ extension field
$\K$ over $\Q$, and use the regular representation of $\K$. We now prove
Lemma~\ref{lem:factoring}, which suggests that the linear $\exists$-singularity
problem is not so easy either.
\begin{proof}[Proof of Lemma~\ref{lem:factoring}]
We consider a special case of Problem~\ref{prob:exist_sing} as follows. Assume
$\cB\leq \M(n, \Q)$ is closed under matrix multiplication, so $\cB$ forms an
algebra over $\Q$. In this setting, Problem~\ref{prob:exist_sing} just asks
whether $\cB$ is not a division algebra. We can even specialize further by
considering $\cB$ being a central simple algebra over $\Q$.

In \cite{Ron87}, R\'onyai considered the problem of testing whether a central
simple algebra over $\Q$ of dimension $4$ is isomorphic to $\M(2, \Q)$. He showed
that assuming the generalized Riemann hypothesis, there is a randomized efficient
reduction from deciding quadratic residuosity modulo squarefree composite
numbers
to this problem. In \cite{Ron87}, the algebras are represented by structural
constants, but these can be turned into matrix representations in $\M(4, \Q)$
(see e.g. \cite{IR99}). It follows that there is an analogous reduction for
matrix algebras.

We can then conclude the proof, because such an
algebra is either isomorphic to $\M(2, \Q)$ (in which there exists a nonzero
singular matrix) or a division algebra (in which every nonzero matrix is
full-rank).
\end{proof}

\section{Proof of Theorem~\ref{thm:iso_bound}}\label{sec:iso_bound}

\subsection{Some basic statistics}

All results in this subsection are either classical or straightforward. We collect
them here, and provide
proofs, partly for
completeness, and partly because we will use some of the arguments here in the
following.

We first recall the following bound on the number of subspaces of $\F_q^n$.

\begin{fact}\label{fact:gaussian}
\begin{enumerate}
\item For $d\leq \N$, $0\leq d\leq n$, the number of dimension-$d$ subspaces of
$\F_q^n$ is equal to
the Gaussian binomial
coefficient
$$\binom{n}{d}_q:=\frac{(q^n-1)\cdot(q^n-q)\cdot\ldots\cdot(q^n-q^{d-1})}
{(q^d-1)\cdot(q^d-q)\cdot\ldots\cdot(q^d-q^{d-1})}.
$$
\item The Gaussian binomial coefficient satisfies:
$$
q^{(n-d)d}\leq \binom{n}{d}_q\leq q^{(n-d)d+d}.
$$
\item The number of subspaces of $\F_q^n$ is $q^{\frac{1}{4}n^2+\Theta(n)}$.
\end{enumerate}
\end{fact}
\begin{proof}
(1) is well-known. For (2),
it is enough to verify that for any prime power $q$, and $n, d, k\in \N$, $n\geq
d>k$, we have
$$
q^{n-d}\leq \frac{q^n-q^k}{q^d-q^k}\leq q^{n-d+1}.
$$
For (3), it is well-known that $\binom{n}{d}_q$ achieves maximal over $d$ at
$d=\lfloor
n/2\rfloor$. So
we have
$$
q^{\frac{1}{4}n^2-\frac{1}{4}}\leq \binom{n}{\lfloor n/2\rfloor}_q\leq
\sum_{d=0}^n\binom{n}{d}_q\leq (n+1)\cdot \binom{n}{\lfloor n/2\rfloor}_q\leq
q^{\frac{1}{4}n^2+\lfloor n/2\rfloor+\log (n+1)},
$$
from which the result follows.
\end{proof}

Analogously, we consider the number of isotropic spaces of a non-degenerate
alternating form $A\in \Lambda(n, q)$.
\begin{proposition}\label{prop:iso_bd}
Let $A\in \Lambda(n, q)$, $n$ even, be a non-degenerate alternating matrix. Then
we
have the following.
\begin{enumerate}
\item For $d\in\N$, $0\leq d\leq n/2$, the number of dimension-$d$
isotropic spaces of $A$ is
$$
I(A, d):=\frac{(q^n-1)\cdot(q^{n-1}-q)\cdot\ldots\cdot(q^{n-(d-1)}-q^{d-1})}
{(q^d-1)\cdot(q^d-q)\cdot\ldots\cdot(q^d-q^{d-1})}.
$$
For $d\in \N$, $d> n/2$, there are no dimension-$d$ isotropic
spaces.
\item For $d\in\N$, $0\leq d\leq  n/2$, $I(A, d)$ is bounded as
follows:
$$
q^{nd-\frac{3}{2}d^2+\frac{1}{2}d}\leq I(A, d)\leq
q^{nd-\frac{3}{2}d^2+\frac{3}{2}d}.
$$
\item The number of isotropic spaces of $A$ is $q^{\frac{1}{6}n^2+\Theta(n)}$.
\item The number of maximal isotropic spaces of $A$ is
$q^{\frac{1}{8}n^2+\Theta(n)}$.
\end{enumerate}
\end{proposition}
\begin{proof}
For (1), suppose we have chosen $u_1, \dots, u_i$ such that $\langle u_1, \dots,
u_i\rangle$ is an isotropic space. We then need to select the next eligible
$u_{i+1}$, such that $\langle u_1, \dots, u_{i+1}\rangle$ forms an isotropic
space. Since $u_{i+1}$ needs to satisfy $u_{i+1}^tAu_j=0$ for $1\leq j\leq i$, and
$A$ is non-degenerate, $u_{i+1}$ should be from a dimension-$(n-i)$ subspace,
namely the subspace orthogonal to $Au_j$, $1\leq j\leq i$. Furthermore, $u_{i+1}$
is not in $\langle u_1, \dots, u_i\rangle$. So there are $q^{n-i}-q^i$ choices of
$u_{i+1}$ in the $i$th step. This explains the numerator. The denominator is of
such form,
because for each isotropic space there are these many ordered bases.

For (2), it follows from the same argument as the proof for
Fact~\ref{fact:gaussian} (2).

For (3), note that $nd-\frac{3}{2}d^2$ achieves its maximum at $d=\frac{1}{3}n$.

For (4), note that maximal isotropic spaces are of dimension $n/2$. This is
because for any isotropic space $U$ of dimension $d<n/2$, we can choose an
eligible
$u_{d+1}$ as in the proof for (1), such that $\langle U, u_{d+1}\rangle$ is also
an isotropic space.
\end{proof}

\begin{proposition}\label{prop:enum}
Let $A\in \Lambda(n, q)$, $n$ even, be a non-degenerate alternating matrix. Then
all isotropic spaces of $A$ can be enumerated in time $q^{\frac{1}{6}n^2+O(n)}$.
\end{proposition}
\begin{proof}
We enumerate isotropic spaces according to dimensions in an increasing order. Each
subspace of $\F_q^n$ is represented by an ordered basis. We will maintain a list
$L$ of all
isotropic subspaces, and for each isotropic space $U$ of dimension $d$, maintain a
list of isotropic spaces of dimension $d+1$ that contain $U$, denoted as $L(U)$.
Note that for a fixed $U$, there are
at most $q^{n-d}$ such spaces. In other words, we will record the lattice of
isotropic spaces.

Suppose
we have enumerated all isotropic spaces of dimensions $\leq d$. To enumerate
isotropic spaces of dimension $d+1$, we maintain a list of such spaces. Then for
each isotropic space $U$ of dimension $d$, and for each $u\in \rad(U)\setminus U$,
we
form $U'=\langle u, U\rangle$, and test whether $U'$ is in $L(U)$. If not, then we
add
$U'$ to $L$. We also add it to $L(U)$, and for every dimension $d$-subspace
$\tilde U$ of $U'$, add $U'$ to
$L(\tilde U)$. Otherwise we move on.

Clearly, in the above procedure, each isotropic space will added, and only added
to $L$ once. This
procedure runs in
time $N\cdot q^{O(n)}$, where $N$ denotes the number
of isotropic spaces of $A$. We can then conclude by resorting to
Proposition~\ref{prop:iso_bd} (3).
\end{proof}

When working with maximal isotropic spaces, it is enough to restrict our attention
to just non-degenerate matrix spaces.
\begin{observation}\label{obs:mis_nondeg}
For $\cA\leq \Lambda(n, \F)$, any maximal isotropic space of $\cA$ contains
$\rad(\cA)$.
\end{observation}

\subsection{A non-trivial upper bound on the number of maximal isotropic spaces}

For $\cA\leq \Lambda(n, q)$, let $\MISet(\cA)$ be the set of maximal isotropic
spaces of $\cA$, and $\MINum(\cA)$ be the number of maximal isotropic spaces of
$\cA$, e.g. the size of $\MISet(\cA)$. Let $\MMINum(n,
q)$ be the maximum of $\MINum(\cA)$ over all $\cA\leq \Lambda(n, q)$. By
Fact~\ref{fact:gaussian} (3) and
Proposition~\ref{prop:iso_bd} (4),
$$q^{\frac{1}{4}n^2+O(n)}\geq \MMINum(n,
q)\geq q^{\frac{1}{8}n^2+\Omega(n)}.$$

\paragraph{Theorem~\ref{thm:iso_bound}, slightly reformulated.}
Let $\MMINum(n, q)$ be as above. Then $\MMINum(n, q)\leq
q^{\frac{1}{6}n^2+C\cdot n}$ for some large enough absolute constant $C$.

Let us illustrate the proof strategy for Theorem~\ref{thm:iso_bound}, before we
enter the details.

The starting point of our proof is the alternative proof bounding the number of
maximal independent sets by Wood \cite{Wood11}.

The core of Wood's argument is the
following. Let $G=(V, E)$ be a graph on $n$ vertices. Recall that we want to prove
that the number of maximal independent sets in $G$ is no more than
$g(n)=3^{\frac{n}{3}}$. We shall do an induction on $n$. Let $v\in V$ be a vertex
of minimal
degree $d$. Let $N(v)$ be the set of neighbours of $v$ together with $v$ (e.g. the
closed neighbourhood of $v$). Then
any maximal independent set $I$ contains some $w\in N(v)$, as
otherwise $I\cup\{v\}$ would be a larger independent set. If $I$ contains $w\in
N(v)$, then $I\setminus \{w\}$ would be a maximal independent set of
$G|_{V\setminus N(w)}$, the induced subgraph of $G$ on $V\setminus N(w)$. Since
$G|_{V\setminus N(w)}$ is of size $\leq n-d-1$, we then have
\begin{equation}\label{eqn:wood}
g(n)\leq (d+1)\cdot g(n-d-1).
\end{equation}
From this relation and the induction hypothesis, the result follows in a rather
straightforward fashion.

In the following, we will develop a linear algebraic analogue of
Equation~\ref{eqn:wood}. However, just applying this does not suffice, when there
are many vectors of degree $1$.

We remedy this by showing that in this setting,
the maximum rank is large, which allows us to use an argument similar to one in
Proposition~\ref{prop:iso_bd}. More specifically,
recall that in
Proposition~\ref{prop:iso_bd} (3), we showed that the number of
isotropic spaces of a non-degenerate alternating matrix is bounded from above by
$q^{\frac{1}{6}n^2+O(n)}$. Note that any maximal isotropic space of $\cA$ is an
isotropic space of any $A\in \cA$. So if $\cA$ contains a non-degenerate $A$, we
can immediately obtain Theorem~\ref{thm:iso_bound} in this case. However, there
are non-degenerate matrix spaces that do not contain non-degenerate alternating
matrices. For example, the following is an alternating matrix space of maximum
rank $2$, written in a parametrized form:
$$
A=\begin{bmatrix}
0 & x_1 & \dots & x_n \\
-x_1 & 0 & \dots & 0 \\
\vdots & \vdots & \ddots & \vdots \\
-x_n & 0 & \dots & 0
\end{bmatrix}.
$$
For our purpose, we will need to bound the number of isotropic spaces for
matrix spaces of rank $> \frac{2}{3}n$. So the following lemma is required;
its proof is postponed to
Section~\ref{subsubsec:technical}.
\begin{lemma}\label{lem:technical}
Let $A\in \Lambda(n, q)$ be of rank $>\frac{2}{3}n$. Then the number of isotropic
spaces of $A$ is bounded from above by $q^{\frac{1}{6}n^2+Dn}$ for some large
enough absolute constant $D$.
\end{lemma}
We are now ready to prove Theorem~\ref{thm:iso_bound}.

\begin{proof}[Proof of Theorem~\ref{thm:iso_bound}]
Let $\cA\leq \Lambda(n, q)$. By Observation~\ref{obs:mis_nondeg}, for our
purpose, we can assume that $\cA$ is non-degenerate. Let
$g_q(n)=q^{\frac{1}{6}n^2+Cn}$. We prove by an induction on $n$. Assume
$\MMINum(\ell, q)\leq g_q(\ell)$ holds for any $\ell<n$. Our goal is to show that
$\MINum(\cA)\leq g_q(n)$.


Let $d=\min\{\deg_\cA(v) :
v\in\F_q^n, v\neq \zerovec\}$. As $\cA$ is non-degenerate, $d\geq 1$. Take any
$v\in \F_q^n$ of degree $d$, and let $c=n-d$ be the codegree of $v$. Let
$N(v):=(\F_q^n\setminus \rad_\cA(v))\cup \{v\}=\{u\in \F_q^n : \exists A\in \cA,
u^tAv\neq 0\}\cup \{v\}$. We call $N(v)$ the closed neighbourhood of $v$. Note
that $|N(v)|=q^n-q^c+1$.

Let $U\leq \F_q^n$ be a maximal isotropic space of $\cA$. Clearly, $U\cap N(v)\neq
\emptyset$. As otherwise, we have $U\subseteq \rad(v)$ and $v\not\in U$. This is
equivalent to that $v\in \rad(U)$ and $v\not\in U$. It follows that $U\subsetneq
\rad(U)$, so by Observation~\ref{obs:isotropic}, $U$ is not maximal, a
contradiction.

Therefore there exists some $w\in N(v)\cap U$. It follows that $U\subseteq
\rad(w)$. Since
$U$ is maximal isotropic in $\cA$, $U$
is also a maximal isotropic space of $\cA|_{\rad(w)}$. As $\deg(w)\geq \deg(v)=d$,
$\dim(\rad(w))\leq c$. Furthermore, note $w$ is an isolated vector in
$\cA|_{\rad(w)}$.
We then have
\begin{eqnarray}
\MINum(\cA) & \leq & \sum_{w\in N(v)}\MINum(\cA|_{\rad(w)}) \label{eqn:minum} \\
& \leq & (q^n-q^c+1)\cdot g_q(c-1), \label{eqn:wood_lin}
\end{eqnarray}
where the second inequality is due to the induction hypothesis.
Note that on the right hand side, we have $g_q(c-1)$ instead of $g_q(c)$, because
$w$ is an isolated vector in $\cA|_{\rad(w)}$, and
Observation~\ref{obs:mis_nondeg}. The reader may want to compare this
with
Equation~\ref{eqn:wood}.

Now suppose $d\geq 2$, that is, $c\leq n-2$. We then have
\begin{eqnarray*}
\MINum(\cA) & \leq & q^n\cdot g_q(n-3)\\
 & \leq & q^n\cdot q^{\frac{1}{6}(n-3)^2+C(n-3)} \\
 & = & q^{\frac{1}{6}n^2+Cn+(\frac{3}{2}-3C)} \\
 & \leq & q^{\frac{1}{6}n^2+Cn}.
\end{eqnarray*}
Note that the second inequality is by the induction hypothesis, and the last
inequality holds as long as $C\geq 1$.

Now suppose $d=1$. In this case, Equation~\ref{eqn:wood_lin} is not enough
for our purpose. We then need the following refinement. Partition $N(v)$ as
$N_1(v)\cup N_{\geq 2}(v)$, where $N_1(v)=\{w\in \F_q^n : w\in N(v), \deg(w)=1\}$,
and $N_{\geq 2}(v)=N(v)\setminus N_1(v)$. A refinement of Equation~\ref{eqn:minum}
gives
that
\begin{equation}\label{eqn:wood_refine}
\MINum(\cA)\leq |N_1(v)|\cdot g_q(n-2)+|N_{\geq 2}(v)|\cdot g_q(n-3).
\end{equation}

If $|N_1(v)|\leq q^{\frac{2}{3}n}$, then we have
\begin{eqnarray*}
\MINum(\cA) & \leq & q^{\frac{2}{3}n}\cdot g_q(n-2) + q^n\cdot g_q(n-3) \\
 & \leq & q^{\frac{2}{3}n}\cdot q^{\frac{1}{6}(n-2)^2+C(n-2)}+q^n\cdot
 q^{\frac{1}{6}(n-3)^2+C(n-3)} \\
 & \leq &
 q^{\frac{1}{6}n^2+Cn+(\frac{2}{3}-2C)}+q^{\frac{1}{6}n^2+Cn+(\frac{3}{2}-3C)} \\
 & \leq &
 q^{\frac{1}{6}n^2+Cn-1}+q^{\frac{1}{6}n^2+Cn-1}\\
 & \leq & q^{\frac{1}{6}n^2+Cn}.
\end{eqnarray*}
Note that the second inequality is by the induction hypothesis, the second to the
last
inequality holds as long as $C\geq 1$.

If $|N_1(v)|> q^{\frac{2}{3}n}$, then we first prove the following.
\begin{claim}\label{claim:rank}
We have $\rk(\cA)>\frac{2}{3}n$.
\end{claim}
\begin{proof}[Proof for Claim~\ref{claim:rank}]
Suppose $\dim(\cA)=m$. Let $s=\lfloor\frac{2}{3}n\rfloor+1$, the smallest integer
larger than $\frac{2}{3}n$. We will show that there exists a linear basis of
some $\tilde \cA$ that is isometric to $\cA$, $\tilde A_1, \dots, \tilde A_m\in
\Lambda(n, q)$, such that
\begin{equation}\label{eqn:hope}
\tilde A_1=\begin{bmatrix}
B_1 & -C_1^t \\
C_1 & D_1
\end{bmatrix}, \tilde A_2=\begin{bmatrix}
\zerovec & \zerovec \\
\zerovec & D_2
\end{bmatrix}, \dots, \tilde A_m=\begin{bmatrix}
\zerovec & \zerovec \\
\zerovec & D_m
\end{bmatrix},\end{equation}
where $B_1\in \Lambda(s, q)$, $C_1\in M((n-s)\times s, q)$, and
$D_i\in \Lambda(n-s, q)$. From this linear basis, it is clear that $\begin{bmatrix}
B_1 \\ C_1
\end{bmatrix}$ is of rank $s>\frac{2}{3}n$, as otherwise $\tilde \cA$ would be
degenerate. It would follow then that $\rk(\cA)=\rk(\tilde\cA)>\frac{2}{3}n$.

We first start with an arbitrary linear
basis of $\cA$, say $ A_1, \dots,  A_m\in \Lambda(n, q)$. Recall that
$v$ is
of degree $1$, and $|N_1(v)|>q^{\frac{2}{3}n}$. For later convenience, rename $v$
as $u_1$. Then there
exist $u_2, u_3, \dots, u_s\in \F_q^n$, such that for $i\geq 2$, $u_i\in N_1(v)$,
and $u_1, \dots, u_s$ are linearly independent. As otherwise, suppose the maximum
number of linearly independent $u_i$'s from $N_1(v)$ we can find is $t<s$. Then
since
$N_1(v)>q^{\frac{2}{3}n}\geq q^t=|\langle u_1, \dots, u_t\rangle|$, we can find
$u_{t+1}\in N_1(v)\setminus \langle u_1, \dots, u_t\rangle$, a contradiction.

We then can arrange a change of basis matrix $T$ whose first $s$ columns are $u_1,
\dots, u_s$. Apply this change of basis matrix $T$ (by $T^t\cdot T$) to $
A_1,
\dots,  A_m$ to get a linear basis $\bar A_1, \dots, \bar A_m$ for $\tilde
\cA=T^t\cA T$. Recall that $e_i$ denotes the $i$th standard basis vector. Since
$u_i$'s are of degree $1$, for any
$i\in [s]$, we have $\tilde\cA(e_i)$ is of dimension $1$.
For $2\leq i\leq s$,
since $u_i\in N(v)$, we have
\begin{equation}\label{eqn:eie1}
e_i^t(\bar A_1, \dots, \bar A_m)e_1\neq (0, \dots,
0).
\end{equation}

Without loss of generality, assume $\bar A_1e_1\neq \zerovec$. As $\tilde\cA(e_1)$
is of dimension $1$, we have for any
$2\leq j\leq m$, $\bar A_je_1=\lambda_j\bar A_1e_1$ for some $\lambda_j\in
\F_q$.
We claim that for
any $2\leq
i\leq s$, the $i$th entry of $\bar A_1e_1$, $\bar A_1e_1(i)\neq 0$. If not, then
for any $2\leq j\leq m$,
$\bar A_je_1(i)=\lambda_j\bar A_1e_1(i)=0$. This is equivalent to say that
$e_i^t(\bar A_1,
\dots, \bar A_m)e_1=0$, contradicting Equation~\ref{eqn:eie1}.

As $\bar A_i$'s are alternating matrices, we have for any $i, j, k$, $(\bar
A_ie_j)(k)=-(\bar
A_ie_k)(j)$. It follows that for $2\leq i\leq s$, $\bar A_1e_i(1)=-\bar
A_1e_1(i)\neq 0$, and for $2\leq j\leq m$, $\bar A_je_i(1)=-\bar
A_je_1(i)=-\lambda_j\bar A_1e_1(i)=\lambda_j\bar A_1e_i(1)$. Since $\cA(e_i)$ is
of dimension $1$ for $2\leq i\leq s$, we infer that for $2\leq j\leq m$ and $2\leq
i\leq s$, $\bar A_je_i=\lambda_j\bar A_1e_i$. We then let $\tilde A_1=\bar A_1$,
and for $2\leq j\leq m$, $\tilde A_j=\bar A_j-\lambda_j\bar A_1$. Clearly, $\tilde
A_1, \dots, \tilde A_m$ still form a basis of $\tilde \cA$, and they are of the
form in Equation~\ref{eqn:hope}. The claim then follows.
\end{proof}

Combining Claim~\ref{claim:rank} and Lemma~\ref{lem:technical}, the proof is
concluded.
\end{proof}

\subsubsection{Proof of Lemma~\ref{lem:technical}}\label{subsubsec:technical}

Let $c$ be the corank of $A$. We then have $c<\frac{1}{3}n$.



Let $(u_1, \dots,
u_d)$ be an ordered basis of an isotropic
space $U$ of $A$ of dimension $d$. For $i\in [d]$, let $U_i=\langle u_1,
\dots, u_i\rangle$, and let $a_i=\dim(A(U_i))$. We also let $U_0=\zerovec$, and
$a_0=0$. Note that $U=U_d$, and we also let
$a=a_d$. We call such an isotropic space of $(d, a)$ type.
Note that $\dim(\langle U_i,
\rad(A)\rangle)=\dim(U_i)+\dim(\rad(A))-\dim(U_i\cap
\rad(A))=\dim(A(U_i))+\dim(\rad(A))=a_i+c$.

After fixing $u_1, \dots, u_i$, a valid $u_{i+1}$ can come from two sources.
\begin{enumerate}
\item If $u_{i+1}\not\in \langle U_i, \rad(A)\rangle$, then since $u_{i+1}$ needs
to
satisfy
$u_{i+1}^tAu_j=0$ for $j=1, \dots, i$, the number of choices of $u_{i+1}$ is upper
bounded by $q^{n-a_i}-q^i$.
\item If $u_{i+1}\in \langle U_i, \rad(A)\rangle$, then the number of choices of
$u_{i+1}$ is upper
bounded by $q^{c+a_i}-q^i$.
\end{enumerate}
So the following indices are important: for $i\in[a]$, let $b_i$ be the smallest
$j\in [d]$ such
that $a_j=\dim(A(U_j))=i$. We then have $0<b_1<b_2<\dots<b_a\leq d$. We also let
$b_0=0$ and $b_{a+1}=d$.
We
call
such an
ordered basis of $(b_1, \dots, b_a)$ type of an isotropic space of $(d, a)$ type.

The number of possible types of an isotropic space is trivially upper bounded by
$n^2$, and
the number of possible types of ordered bases of isotropic spaces of type $(d, a)$
is upper
bounded by
$\binom{d}{a}\leq 2^d\leq 2^n$. So by a multiplicative factor of $n^2\cdot 2^n$,
we can restrict to consider ordered basis $(u_1, \dots, u_d)$ of a fixed
type $\vec{b}=(b_1, \dots,
b_a)$. By the discussion above, if $j=b_i$, then the number of choices for $u_j$
is upper bounded by $t_{d, a, \vec{b}}(j):=q^{n-(i-1)}-q^{j-1}$. If
$b_i<j<b_{i+1}$, the
number of
choices for $u_j$ is upper bounded by $t_{d, a, \vec{b}}(j):=q^{c+i}-q^{j-1}$.
Recall that
$(q^n-q^i)/(q^d-q^i)\leq q^{n-d+1}$, for any $q$
and $i< d\leq n$.  If $j=b_i$, we have
\begin{equation}\label{eq:jump}
t_{d, a, \vec{b}}(j)/(q^d-q^{j-1})\leq
q\cdot q^{n-(i-1)-d}.
\end{equation}
If $b_i<j<b_{i+1}$, we have
\begin{equation}\label{eq:keep}
t_{d, a,
\vec{b}}(j)/(q^d-q^{j-1})\leq q\cdot q^{(c+i)-d}\leq q\cdot q^{(c+a)-d}
\end{equation}

Each dimension-$d$ subspace of $\F_q^n$ has $(q^d-1)(q^d-q)\dots(q^d-q^{d-1})$
ordered bases, and each ordered basis of an isotropic space of type $(d, a)$ is of
a particular type.
The number of dimension-$d$ isotropic spaces of type $(d,
a)$ can be upper bounded by
\begin{eqnarray*}
& & \sum_{\text{type }\vec{b}=(b_1, \dots, b_a)}\frac{t_{d, a,
\vec{b}}(1)\cdot\ldots\cdot
t_{d, a, \vec{b}}(d)}{(q^d-1)\cdot\ldots\cdot (q^d-q^{d-1})} \\
&= & \sum_{\text{type }\vec{b}=(b_1, \dots, b_a)}\frac{t_{d, a,
\vec{b}}(1)}{q^d-1}\cdot\ldots \cdot \frac{t_{d, a,
\vec{b}}(j)}{q^d-q^{j-1}}\cdot\ldots\cdot \frac{t_{d, a, \vec{b}}(d)}{q^d-q^{d-1}}
\\
&\leq & \sum_{\text{type }\vec{b}=(b_1, \dots, b_a)} q^d\cdot
q^{na-\sum_{i=1}^a(i-1)-da}\cdot q^{(c+a)(d-a)-d(d-a)}\\
&\leq & 2^n\cdot q^{na-a^2/2+(c+a)(d-a)-d^2+d+a/2}.
\end{eqnarray*}
Let us explain the first inequality. The $q^d$ term is because of the $q$ terms on
the right hand sides of Equations~\ref{eq:jump} and~\ref{eq:keep}. The
$q^{na-\sum_{i=1}^a(i-1)-da}$ is by collecting those terms from
Equation~\ref{eq:jump}, and the $q^{(c+a)(d-a)-d(d-a)}$ term is by collecting
those terms from Equation~\ref{eq:keep}.

It is then clear that we need to bound $f(n, d, a)=na-a^2/2+(c+a)(d-a)-d^2$ for
$1\leq
a\leq d\leq n$. After some arrangement, we have
$$
f(n, d, a)=-\frac{3}{2}(a-\frac{1}{3}(n+d-c))^2+\frac{1}{6}(n+d-c)^2-d^2+dc.
$$
We then distinguish between two cases.
\begin{enumerate}
\item Case (i): when $\frac{1}{3}(n+d-c)\leq d$ holds, namely $d\geq (n-c)/2$.
Only in this case, $a$ can be set to
$\frac{1}{3}(n+d-c)$, and the maximum can be set to $g(n,
d):=\frac{1}{6}(n+d-c)^2-d^2+dc$. After some arrangement, we have
$$g(n,
d)=-\frac{5}{6}(d-\frac{1}{5}(n+2c))^2+\frac{1}{30}(n+2c)^2+\frac{1}{6}(n-c)^2.$$
Since $c<n/3$, we have $(n-c)/2>(n+2c)/5$. Recall that $d\geq (n-c)/2$. So $g(n,
d)$ achieves maximal at $d=(n-c)/2$. Plugging this in, the maximal value is
$$
h(n):=g(n, (n-c)/2)=-\frac{3}{8}(c-\frac{1}{3}n)^2+\frac{1}{6}n^2<\frac{1}{6}n^2.$$
\item Case (ii): when $\frac{1}{3}(n+d-c)>d$ holds, namely $d<(n-c)/2$. In this
case, $f(n, d, a)$ achieves the maximal value at $a=d$, and
$$
f(n, d, d)=-\frac{3}{2}(d-\frac{1}{3}n)^2+\frac{1}{6}n^2\leq \frac{1}{6}n^2,$$
where the inequality becomes an equality at $d=n/3$.
\end{enumerate}
Since in both cases, the maximal value is no more than $\frac{1}{6}n^2$, we can
then conclude the proof.

\subsection{Turning Theorem~\ref{thm:iso_bound} into an algorithm}

The proof of Theorem~\ref{thm:iso_bound} can be turned into an algorithm for
enumerating all maximal isotropic spaces in time $q^{\frac{1}{6}n^2+O(n)}$. We
briefly indicate some algorithmic issues for doing this. Firstly, note that in
time $q^{O(n)}$, one can compute $\deg_\cA(v)$ for all $v\in \F_q^n$. Secondly,
the Equation~\ref{eqn:minum} naturally suggests a recursive algorithm structure.
In the cases when $d\geq 2$, or $d=1$ and $|N_1(v)|\leq q^{\frac{2}{3}n}$, this
recursive structure readily gives the desired algorithm. If $d=1$ and $|N_1(v)|>
q^{\frac{2}{3}n}$, we need to make the proofs of Claim~\ref{claim:rank} and
Lemma~\ref{lem:technical} constructive. Then for each isotropic space of some
$A\in \Lambda(n, \F)$, $A$ of rank $>\frac{2}{3}n$, test whether it is maximal
using Observation~\ref{obs:isotropic}.

For Claim~\ref{claim:rank}, note that the selection of $u_i$'s from $N_1(v)$ can
be done easily
in a greedy way. Other steps are readily constructive. For
Lemma~\ref{lem:technical}, we use the same procedure as in
Proposition~\ref{prop:enum}, whose running time is bounded in time
$q^{\frac{1}{6}n^2+O(n)}$ by Lemma~\ref{lem:technical}.

We then have the following.
\begin{corollary}\label{cor:iso_bound}
Given $\cA\leq \Lambda(n, q)$, all maximal isotropic spaces can be enumerated in
time $q^{\frac{1}{6}n^2+O(n)}$.
\end{corollary}

\section{Proof of Theorem~\ref{thm:lawler}}\label{sec:lawler}

\paragraph{Review of Lawler's algorithm.} We first review Lawler's dynamic
programming idea for computing the chromatic
number \cite{Law76}, and then adapt that idea to our problem.

Given a graph $G=(V, E)$, Lawler's algorithm for computing $\chi(G)$ goes as
follows. The idea is to build a table storing $\chi(H)$ for every induced subgraph
$H$ of $G$. Note that this table is of size $2^n$. To fill in this table,
the starting point is the empty graph with
chromatic number $0$. Suppose we have computed the chromatic numbers of those
induced subgraphs of size $<\ell$. Let $H=(U, F)$ be an induced subgraph of size
$\ell$. Then the chromatic number of $H$ can be computed by the following formula:
$$
\chi(H)=1+\min_{I\subseteq U}\{\chi(H[U\setminus
I])\},
$$
where $I$ goes over all maximal independent sets of $H$, and $H[U\setminus I]$ is
the induced subgraph of $H$ restricting to vertex set
$U\setminus I$. Since there are at most $3^{\ell/3}$ maximal independent sets of
$H$ and they can be enumerated in time $O(3^{\ell/3}\cdot n)$, the exponential
part of the time complexity of this algorithm is
$\sum_{\ell=0}^n\binom{n}{\ell}\cdot 3^{\ell/3}=(1+\sqrt[3]{3})^n$.

\paragraph{Directly applying Lawler's idea to isotropic numbers.} The above idea
can be adapted to compute $\chi(\cA)$ for $\cA\leq \Lambda(n, q)$
as follows. To start with, recall that in the above algorithm we used the
following simple
fact: if a graph $G$ admits a vertex $c$-coloring, then there is a vertex
$c$-coloring in which one part is a maximal independent set. We leave the
reader to check that the analogue of this fact in the alternating matrix space
setting also holds.

Given $\cA\leq \Lambda(n, q)$, we also store a table storing $\chi(\cB)$ for every
induced alternating matrix spaces $\cB$ of $\cA$. Note that this table is of size
$q^{\frac{1}{4}n^2+O(n)}$. To fill in this table, the starting point is the zero
space with isotropic decomposition number $0$. Suppose we have computed the
isotropic decomposition numbers of those induced alternating matrix spaces of
dimension $<\ell$. Let $\cB\leq \Lambda(\ell, q)$ be an induced alternating matrix
space corresponding to $U\leq \F_q^n$ of dimension $\ell$. Then the isotropic
decomposition number of $\cB$ can be computed by the following formula
$$
\chi(\cB)=1+\min_{V\leq U, W\leq U}\{\chi(\cB|_W)\},
$$
where $V$ goes over all maximal isotropic spaces of $\cB$, and $W$ goes over all
complement subspaces of $V$ in $U$. Note that here we also need to enumerate all
complements of $V$, while in the graph setting, the complement set is unique.
Recall that by Theorem~\ref{thm:iso_bound}, there are at most
$q^{\frac{1}{6}\ell^2+O(\ell)}$ maximal isotropic spaces of $\cB$, and they can be
enumerated in time $q^{\frac{1}{6}\ell^2+O(\ell)}$. This gives a bound on the
number of $V$. We bound the number of $W$ using the trivial
$q^{\frac{1}{4}\ell^2+O(\ell)}$ bound. Note that we will need to test whether $W$
is a complement of $V$, which can be done easily. Another more efficient approach
would be to
enumerate all complements of $U$ in time $q^{d(n-d)}\cdot \poly(n, \log q)$ (see
\cite[Proposition 17 in the arXiv version]{LQ17}).

So to fill in those entries corresponding
to alternating matrix spaces induced by $\ell$-dimensional subspaces, the
time complexity is bounded by
\begin{eqnarray*}
& & \binom{n}{\ell}_q\cdot q^{\frac{1}{6}\ell^2+O(\ell)}\cdot
q^{\frac{1}{4}\ell^2+O(\ell)} \\
& \leq & q^{\ell(n-\ell)+\ell}\cdot q^{\frac{1}{6}\ell^2+O(\ell)}\cdot
q^{\frac{1}{4}\ell^2+O(\ell)} \\
& = & q^{\ell n-\frac{7}{12}\ell^2+O(\ell)}\\
& = & q^{-\frac{7}{12}(\ell-\frac{6}{7}n)^2+\frac{3}{7}n^2+O(\ell)}\\
& \leq & q^{\frac{3}{7}n^2+O(n)}.
\end{eqnarray*}
Summing over $\ell\in\{0, 1, \dots, n\}$, we see that the total time complexity is
also bounded by $q^{\frac{3}{7}n^2+O(n)}$.

\paragraph{A new dynamic programming scheme.} In the above, we see that directly
following Lawler's dynamic programming scheme does lead to an improved algorithm
for computing the isotropic decomposition number. However, a key difference with
the classical setting, namely the magnitude of complement subspaces, impacts the
analysis. In the following, we shall use another dynamic programming scheme, still
combined with the $q^{\frac{1}{6}n^2+O(n)}$ upper bound on the number of maximal
isotropic spaces, to achieve the $q^{\frac{5}{12}n^2+O(n)}$ running time as
promised in Theorem~\ref{thm:lawler}.

To do that, we first make a simple observation.
\begin{observation}\label{obs:maximal_chi}
Let $\cA\leq \Lambda(n, \F)$. Then $\chi(\cA)\leq k$, if and only if, there exist
$k$
maximal isotropic spaces $U_1, \dots, U_k$, such that $\F^n=\langle
\cup_{i\in[k]}U_i\rangle$.
\end{observation}
\begin{proof}
For the only if direction, recall that every isotropic space is contained in a
maximal one.
For the if direction, note that from $U_1, \dots, U_k$, we can construct $U_1',
\dots, U_k'$, such that $U_i'\leq U_i$, and $U_1', \dots, U_k'$ form a direct sum
decomposition of
$\F^n$. This shows that $\chi(\cA)\leq k$.
\end{proof}

The key to our algorithm is the following function. For $k\in[n]$ and $W\leq
\F_q^n$, let $f(k, U)$ be the boolean function such that $f(k, W)=1$ if and only
if $W=\langle \cup_{i\in[k]} U_i : U_i\text{ maximal isotropic}\rangle$. For
example, $f(1, W)=1$ if and only if $W$ is a maximal isotropic space.

The following is then a dynamic programming scheme computing $f(k, W)$ for every
$k\in[n]$ and $W\leq \F_q^n$. Let $\cA\leq \Lambda(n, q)$ be an alternating matrix
space. We assume that $\chi(\cA)>1$, as $\chi(\cA)=1$ if and only if $\cA$ is the
zero space.
\begin{enumerate}
\item Use Corollary~\ref{cor:iso_bound} to compute the set of maximal isotropic
spaces of $\cA$, and let $MI$ be this set.
\item Build a table $f$, indexed by $(k, W)$ where $k\in[n]$ and $W\leq \F_q^n$,
and initiate $f(k, W)=0$ for every $k$ and $W$.
\item For every $W\leq \F_q^n$, do:
\begin{enumerate}
\item If $W$ is maximal isotropic, then $f(1, W)=1$.
\end{enumerate}
\item For $k=2, \dots, n$, do:
\begin{enumerate}
\item For every $W\leq \F_q^n$ and every $T\in MI$, do:
\begin{enumerate}
\item If $f(k-1, W)=1$, then let $U=\langle W\cup T\rangle$, and set $f(k, U)=1$.
\item If $U=\F_q^n$, then return ``$\chi(\cA)=k$.''
\end{enumerate}
\end{enumerate}
\end{enumerate}

To prove the correctness of the algorithm, we first note that by induction, the
algorithm correctly computes $f(k, W)$ for every $k$ and $W$. Then
suppose the algorithm
returns with reporting that $\chi(\cA)=k$. Note that in this case, it does find
$k$ maximal isotropic subspaces covering the
whole space $\F_q^n$. So by Observation~\ref{obs:maximal_chi}, $\chi(\cA)\leq k$.
So we are left to show that $\chi(\cA)\geq k$. By way of contradiction, suppose
$\chi(\cA)=k'<k$, so by Observation~\ref{obs:maximal_chi}, there exist maximal
isotropic subspaces
$U_1, \dots, U_{k'}$ that cover $\F_q^n$. Let $W=\langle U_2\cup\dots \cup
U_{k'}\rangle$. Then $W$ is a proper subspace of $\F_q^n$, as otherwise by
Observation~\ref{obs:maximal_chi}
$\chi(\cA)\leq k'-1<k'$, contradicting that $\chi(\cA)=k'$. But this means that
$f(k'-1, W)=1$, so in Step (4.a), when enumerating $W$ and $T=U_1$, the algorithm
would go through steps (4.a.i) and (4.a.ii), and outputs that $\chi(\cA)=k'$. This
gives us the desired contradiction.

To estimate the running time of the algorithm, note that Step (1) costs
$q^{\frac{1}{6}n^2+O(n)}$ by Corollary~\ref{cor:iso_bound}. All subspaces can be
enumerated in time $q^{\frac{1}{4}n^2+O(n)}$ by the same technique as in the proof
of Proposition~\ref{prop:enum}. The total running time is then dominated by the
loop in steps (4) and (4.a), which is $n\cdot q^{\frac{1}{6}n^2+O(n)}\cdot
q^{\frac{1}{4}n^2+O(n)}=q^{\frac{5}{12}n^2+O(n)}$. This concludes the proof of
Theorem~\ref{thm:lawler}.

\section{Proofs for Propositions~\ref{thm:maximum_isC} 
and~\ref{prop:bounded_deg}}\label{sec:maximum_isC}

\begin{proof}[Proof of Proposition~\ref{thm:maximum_isC}]
When $\F=\C$ and the input instance is over $\Z$, we shall formulate the maximum
isotropic space problem as a problem about the solvability of a system of
integeral polynomial equations over $\C$. The result would follow then by using
Koiran's
result that the Hilbert Nullstellensatz is in PH, assuming the generalized Riemann
hypothesis \cite{Koi96}. We first cite
Koiran's result as follows, following the formulation of \cite[Theorem
2.10]{Mul17}.

\begin{theorem}[{\cite{Koi96}}]\label{thm:koiran}
The problem Hilbert's Nullstellensatz of deciding whether a given system of
multivariate integral polynomials, specified as arithmetic circuits, has a
solution over $\C$
is in PSPACE unconditionally, and in 
$RP^{NP}\subseteq \Pi_2$ assuming
the generalized Riemann hypothesis.
\end{theorem}

Therefore, to put the maximum isotropic space problem over $\C$ in PSPACE
unconditionally, and in PH assuming the generalized Riemann hypothesis, for 
instances
given by integral alternating matrices, we only need to formulate this problem as
deciding the solvability of a system of integral polynomials represented by
arithmetic circuits. This can be done as follows. Suppose we are given
$\cA=\langle A_1, \dots, A_m\rangle\leq \Lambda(n, \C)$ where $A_i$'s are integral
matrices, and we want to know whether there exists an isotropic space of dimension
$d$ for $\cA$. Then $\cA$ has an isotropic space of dimension $d$ if and only if
there exists an invertible matrix $T$ such that for any $i\in[m]$, the left-upper
$d\times d$ submatrix of $T^tA_iT$ consists of all zero entries. We then set up an
$n\times n$ variable matrix $X=(x_{i,j})_{i,j\in[n]}$, and a variable $y$. For
every $i\in[m]$, set the entries of the left-upper $d\times d$ submatrix of
$X^tA_iX$ to be zero. This gives $md^2$ integral quadratic polynomials in
$x_{i,j}$'s.
To enforce that the valid solutions are from invertible matrices, we set up the
equation $\det(X)\cdot y=1$, which is also an integral polynomial. Note that the
polynomial $\det(X)$ can be expressed
as a small arithmetic circuit. It is straightforward to verify that these 
$(md^2+1)$
equations in $x_{i,j}$ and $y$ have a non-trivial solution if and only if $\cA$
has a dimension-$d$ isotropic space.

For isotropic $3$-decomposition problem, the idea is basically the same. The only 
small complication is that we need to specify the dimensions of the three 
isotropic spaces in a $3$-isotropic decomposition. But the number of possibilities 
is at most $n^3$, which we can enumerate. After fixing some $(d_1, d_2, d_3)$, 
where $d_i\in\Z^+$, $n\geq d_1\geq d_2\geq d_3\geq 1$, and $d_1+d_2+d_3=n$, we can 
construct a system of integral polynomial equations to express the condition that 
there exists a $3$-isotropic decomposition with these dimensions, just as in the 
case of the maximum isotropic space problem. This concludes the proof.
\end{proof}

\begin{proof}[Proof of Proposition~\ref{prop:bounded_deg}]
Recall that we have $\cA\leq \Lambda(n, \F)$, and our goal is to prove
$\chi(\cA)\leq O(\Delta(\cA)\cdot \log n)$. Here, $\Delta(\cA):=\max\{\deg_\cA(v)
: v\in \F^n\}$, and $\deg_\cA(v):=\dim(\langle Av : A\in \cA\rangle)$. We will
also use a greedy algorithm to construct an isotropic $C$-decomposition with
$C\leq O(\Delta(\cA)\cdot \log n)$.

For $v\in \F^n$, recall that $\rad_\cA(v)=\{ u\in \F^n : u^tAv=0\}$. Consider the
following algorithm.
\begin{enumerate}
\item Set $k=0$, and $U=\zerovec$;
\item While $\dim(U)<n$, do:
\begin{enumerate}
\item Set $k=k+1$;
\item Let $W$ be any complementary subspace of $U$;
\item Let $S=\emptyset$;
\item While $\dim(W)>|S|$, do:
\begin{enumerate}
\item Take any $w\in W\setminus \langle S\rangle$; // $\langle
\emptyset\rangle:=\zerovec$
\item $S=S\cup w$;
\item $W=W\cap \rad_\cA(w)$;
\end{enumerate}
\item Let $U_k=\langle S\rangle$;
\item $U=\langle U\cup U_k\rangle$;
\end{enumerate}
\item Return $U_1\oplus U_2\oplus \dots\oplus U_k$.
\end{enumerate}

We first argue that $U_1\oplus U_2\oplus \dots\oplus U_k$ is an isotropic
$k$-decomposition of $\cA$. To see this, we note that because of Step (2.d.iii),
the condition $W\subseteq \rad(\langle S\rangle)$ holds in the loop of Step (d),
so
$\langle S\rangle$ maintains as an isotropic space by
Observation~\ref{obs:isotropic}.

We then show that $k=O(\Delta(\cA)\cdot \log n)$ when the algorithm terminates.
For this, let $d_i=\dim(U_i)$, and
$D_i=d_1+\dots+d_i$. Observe that $\dim(\rad_\cA(w))\geq n-\Delta(\cA)$. It
follows that
$\dim(W\cap \rad_\cA(w))=\dim(W)+\dim(\rad_\cA(w))-\dim(\langle W\cup
\rad_\cA(w)\rangle)\geq \dim(W)+(n-\Delta(\cA))-\dim(\langle W\cup
\rad_\cA(w)\rangle)\geq \dim(W)-\Delta(\cA)$. Therefore, in the computing
procedure of $U_i$, we have $d_i=\dim(U_i)\geq \frac{n-D_{i-1}}{\Delta(\cA)}$.
This implies that $n-D_i=n-D_{i-1}-d_i\leq (n-D_{i-1})(1-1/\Delta(\cA))$.
Therefore, adding a new $U_i$ to the direct sum decomposition decreases the value
of $n-D_i$ by a factor of at least $1-1/\Delta(\cA)$. Therefore the algorithm
terminates in at most $\log_{1-1/\Delta(\cA)}\frac{1}{n}=O(\Delta(\cA)\cdot \log
n)$ steps.
\end{proof}

\section{Proofs of theorems~\ref{thm:large_abelian} and~\ref{thm:maximal_abel}
}\label{sec:application}

In this section we prove theorems~\ref{thm:large_abelian} 
and~\ref{thm:maximal_abel}. While the proofs are straightforward for experts, we 
include details 
for completeness. We shall refer to some facts in Section~\ref{app:origin} from 
time to time. 

\begin{proof}[Proof of Theorem~\ref{thm:large_abelian}]
Recall that the goal is to show that deciding whether a matrix group contains an
abelian subgroup of order $\geq s$ is NP-hard for some $s\in\N$. We shall reduce 
the 
maximum
isotropic
space problem, which is NP-hard by Corollary~\ref{cor:reduction}, to this problem.

For this we shall need the following classical construction\footnote{We thank
James B. Wilson for communicated this construction to us.}.
Let $p$ be an odd prime, and let $\cA\leq \Lambda(n, p)$ be given by an ordered 
linear
basis
$\bA=(A_1, \dots, A_m)$. Recall that $e_i$ denotes the $i$th standard 
basis vector. From $\bA$, for $i\in[n]$, construct $B_i=[A_1e_i, \dots,
A_me_i]\in\M(n\times m, p)$. That is, the $j$th column of $B_i$ is the
$i$th
column of $A_j$. Then for $i\in[n]$, construct
$$\tilde B_i
=
\begin{bmatrix}
1 & e_i^t & 0 \\
0 & I_n & B_i \\
0 & 0 &  I_m
\end{bmatrix}\in \GL(1+n+m, p),
$$
and for $j\in[m]$, construct
$$
\tilde C_j
=
\begin{bmatrix}
1 & 0 & e_j^t \\
0 & I_n & 0 \\
0 & 0 &  I_m
\end{bmatrix}\in \GL(1+n+m, p).
$$

Let $G_\bA$ be the matrix group generated by $\tilde B_i$ and $\tilde C_j$. Then
it can be verified easily that,
$G_\bA$ is isomorphic to the Baer group (see Section~\ref{app:origin})
corresponding to the alternating bilinear map defined by $\bA$ (see
Appendix~\ref{app:spaces_maps}). In particular, $[G, G]\cong \Z_p^m$, and $G/[G, 
G]\cong \Z_p^n$. 
By the correspondence between isotropic
spaces
of $\cA$ and abelian normal subgroups of $G_\bA$ containing the commutator
subgroup (see Section~\ref{app:origin}),
deciding whether $\cA$ has an isotropic
space of dimension $\geq d$ is equivalent to deciding whether $G_\bA$ has an
abelian subgroup of order $\geq s=p^{m+d}$. This completes the reduction.
\end{proof}

\begin{proof}[Proof of Theorem~\ref{thm:maximal_abel}]
Let $P$ be a $p$-group of class $2$ and exponent $p$, and let $\phi:P/[P,P]\times 
P/[P,P]\to [P,P]$ be the commutator map. 
The proof of Theorem~\ref{thm:maximal_abel} basically follows from the 
correspondence 
between abelian subgroup containing $[P,P]$ and isotropic spaces of $\phi$ as 
described in Section~\ref{app:origin}. The 
only small caveat here is that we need a bound on the dimension of $P/Z(P)$ 
instead of the dimension of $P/[P,P]$. To overcome this, we first observe that a 
maximal abelian subgroup of 
$P$ necessarily contains the center $Z(P)$, which in turn contains $[P,P]$ by the 
class-$2$ condition. Then we only need to note that 
$Z(P)/[P,P]$ corresponds to the radical of $\phi$, and recall that the number of 
maximal isotropic spaces only depends on the non-degenerate part of $\phi$ by 
Observation~\ref{obs:mis_nondeg}. The proof then can be concluded.
\end{proof}

\section{A quantum variant of the theory}\label{sec:quantum}

One way to extend isotropic spaces and isotropic decompositions to the 
quantum information setting is as follows. Briefly speaking, firstly we restrict 
to the complex number field $\C$. Secondly, 
instead of tuples of alternating matrices, we will consider tuples 
of matrices which represent a so-called irreducible quantum channels. Thirdly, 
instead of general linear groups, we will consider unitary groups. 

For detailed explanations, we need some notation. For $A\in \M(n, \C)$ we use 
$A\succeq 0$ to denote that $A$ is positive semi-definite, and $A\succ 0$ to 
denote that $A$ is positive definite. For $\sB=\{B_1, \dots, B_m\}\subseteq \M(n, 
\C)$, we 
let $\tilde{\sB}:\M(n, \C)\to \M(n, \C)$ be the function sending $A\in \M(n, \C)$ 
to $\sum_{i=1}^mB_iAB_i^\dagger$. It is clear that $\tilde{\sB}$ can be 
represented as an $n^2\times n^2$ matrix $\sum_{i=1}^mB_i\otimes B_i^*$, where 
$B_i^*$ stands for the entry-wise complex conjugation of $B_i$. 

Let 
$\D(n, \C)\subseteq \M(n, \C)$ be the set of 
$n\times n$ semi-positive definite matrices with unit trace over $\C$, and let 
$\D^+(n, \C)\subseteq \D(n, \C)$ consist of those positive definite matrices in 
$\D(n, \C)$. Elements from $\D(n, \C)$ are known as quantum states. 

Let $\QC(n, 
\C)$ be the set of sets of matrices $\sB=\{B_1, \dots, B_m\}\subseteq  
\M(n, 
\C)$ satisfying $\sum_{i=1}^mB_i^\dagger B_i=I$. Functions of the form 
$\tilde{\sB}$ 
for $\sB\in \QC(n, \C)$ are known as quantum channels, as they are completely 
positive and trace preserving. 

We then define isotropic spaces and decompositions in the quantum setting. To 
define 
isotropic spaces, we essentially follow the same pattern as in the alternating 
matrix 
space 
setting. For isotropic decompositions, we shall require that the direct sum 
decomposition is also an orthogonal one, as the underlying spaces of quantum 
channels are Hilbert 
spaces which come with a norm. 
\begin{definition}
Let $\sB=\{B_1, \dots, B_m\}\in\QC(n, 
\C)$. An isotropic space of $\sB$ is a subspace $U\leq \C^n$, such that for any 
$u, u'\in U$, and any $B_i$, we have $u^\dagger B_iu'=0$. An isotropic 
$c$-decomposition 
of $\sB$ is an orthogonal direct sum decomposition of $\C^n=U_1\oplus U_2\oplus 
\dots\oplus U_c$ such that each $U_i$ is a non-zero isotropic space of $\sB$. 
\end{definition}


\subsection{From connected graphs to irreducible quantum channels}

In this subsection, we establish a connection between independent sets and vertex 
colorings of connected graphs, and isotropic spaces and decompositions of a 
particular type of quantum channels, called irreducible channels (defined below). 

We obtain two main results. The first 
result, Proposition~\ref{prop:reduction}, reduces certain problems for connected 
graphs to the corresponding ones 
for irreducible quantum channels. This result corresponds to 
Theorem~\ref{thm:reduction}. 
The second result, Theorem~\ref{thm:q_2_decomp}, gives an efficient algorithm for 
isotropic $2$-decomposition in this setting. This result corresponds to 
Theorem~\ref{thm:2-decomp}, but the techniques are completely different. 

Let $\IQC(n, \C)\subseteq \QC(n, \C)$ consist of those $\sB\in \QC(n, \C)$ 
satisfying the following: there exists a unique fixed $\rho\in \D(n, \C)$ of $\tilde{\sB}$, and further $\rho\in \D^+(n, \C)$, where $\rho$ is said to be fixed of $\tilde{\sB}$ if 
$\tilde{\sB}(\rho)=\rho$. Such $\sB$ and 
$\tilde{\sB}$ are called \emph{irreducible}. Irreducible quantum channels have 
been studied in e.g. \cite{davies1970quantum} and \cite[Sec. 
6.2]{wolf2012quantum}. In 
particular, the definition of irreducible quantum channels follows from 
\cite[Theorem 13]{davies1970quantum}. Furthermore, given $\sB\in \QC(n, \C)$, let 
$M$ be the $n^2\times n^2$ matrix representation of $\tilde\sB$. Then $\sB\in 
\IQC(n, \C)$ if 
and only if the algebraic and geometric multiplicities of the eigenvalue $1$ of 
$M$ are both $1$, and any $1$-eigenvector is of full-rank. 

We first observe that a simple and connected graph can be realized as an 
irreducible 
quantum 
channel as follows. This is classical, but for completeness we spell out the 
details. Let $G=([n], E)$ be a connected graph. For each $i\in [n]$, 
let $d_i$ be the degree of $i$. We construct the following set of matrices 
$\sB_G=\{\frac{1}{\sqrt{d_j}}\cdot E_{i,j}, \frac{1}{\sqrt{d_i}}\cdot 
E_{j,i} : \{i, j\}\in E\}$. Note that $|\sB_G|=2|E|$. 
\begin{proposition}
Let $G$ and $\sB_G$ be as above. Then $\sB_G\in 
\IQC(n, \C)$.
\end{proposition}
\begin{proof}
We first verify that $\tilde\sB_G$ 
is a
quantum channel. For this, observe that $(\frac{1}{\sqrt{d_j}}\cdot 
E_{i,j})^\dagger \frac{1}{\sqrt{d_j}}\cdot E_{i,j}=\frac{1}{d_j}E_{j,i}E_{i,j}=\frac{1}{d_j}E_{j,j}$. Since
each vertex $i$ connecting to $j$ contributes one such term, and it follows that
$\sum_{E\in \sB_G}E^\dagger E=I$. We then verify that $\tilde\sB_G$ is 
irreducible. 
For this, consider $P=(p_{i,j})$ where $p_{i,j}=1/d_i$, which represents the 
transition matrix of the Markov chain naturally associated with $G$. Since $G$ is 
connected, this Markov chain is irreducible, so there exists a unique probability 
distribution, e.g. a row vector $s=(s_1, \dots, s_n)$ satisfying $s_i>0$, 
$\sum_is_i=1$, such that $sP=s$ (see e.g. \cite[Corollary 1.17]{levin2017markov}). 
It can then be verified that the matrix 
$S=\mathrm{diag}(s_1, \dots, s_n)\in \D^+(n, \C)$ is fixed by $\tilde\sB_G$. 
To see that this is the unique fixed state, we represent $\tilde\sB_G$ as an 
$n^2\times n^2$ matrix $M_G$. It is not hard to see that by conjugating with a 
permutation matrix, $M_G$ is of the form $\begin{bmatrix}P & \zerovec \\ \zerovec 
& \zerovec\end{bmatrix}$. Therefore, the algebraic and geometric 
multiplicities of the eigenvalue $1$ of $M_G$ are the same as those for $P$, which 
are $1$ by the Perron-Frobenius theory. It follows that $\sB$ is 
irreducible. 
\end{proof}

\begin{proposition}\label{prop:reduction}
Let $G=([n], E)$ be a connected graph, and let $\sB_G\in \IQC(n, \C)$ be as above. 
\begin{enumerate}
\item $G$ has a size-$s$ independent set if and only if $\sB_G$ has a 
dimension-$s$ isotropic space;
\item $G$ has a vertex $c$-coloring if and only if $\sB_G$ has an isotropic 
$c$-decomposition.
\end{enumerate}
\end{proposition}
\begin{proof}
(1) The only if direction is trivial. For the if direction, let $U$ be a 
dimension-$s$
isotropic space of $\sB_G$. Then $U$ is also a dimension-$s$ isotropic space of 
the 
alternating matrix space $\langle 
E_{i,j}-E_{j,i} : \{i,j\}\in E\rangle$, because it is a subspace of $\langle 
\sB_G\rangle$. We can then conclude by resorting to  
Theorem~\ref{thm:reduction}. 

(2) The only if direction is trivial; observe that the direct sum decomposition 
obtained from a vertex coloring as in Theorem~\ref{thm:reduction} is also an 
orthogonal direct sum decomposition. For the if direction, we observe that, an 
orthogonal direct sum decomposition into isotropic spaces for $\sB_G$ is also one 
for the alternating matrix space $\langle 
E_{i,j}-E_{j,i} : \{i,j\}\in E\rangle$. We can then conclude by resorting to  
Theorem~\ref{thm:reduction}. 
\end{proof}

Since the maximum independent set problem and the vertex $3$-coloring problem on 
connected graphs are also NP-hard, we have the following. 
\begin{corollary}
The maximum isotropic space problem and the isotropic $3$-decomposition problem 
for $\sB\in \IQC(n, \C)$
are NP-hard. 
\end{corollary}

This also leaves the isotropic $2$-decomposition problem an interesting question. 
For this, we can resort to the techniques developed for quantum Markov chains, 
mostly notably, based on recent works of periodicity 
of quantum channels~\cite{guan2018decomposition}. 



\begin{theorem}\label{thm:q_2_decomp}
Suppose we are given $\sB\in \IQC(n, \C)$ such that every matrix in $\sB$ are 
over $\Q$. There exists an algorithm that decides whether $\sB$ 
admits an isotropic $2$-decomposition in polynomial time. 
\end{theorem}
\begin{proof}
The key observation is to characterize isotropic 
$2$-decompositions using  the \emph{periodicity} of 
irreducible quantum channels.
\begin{definition}[\cite{fagnola2009irreducible}]
Given $\sB\in \IQC(n, \C)$, the \emph{period} of $\sB$ is the maximum integer $m$ 
for 
which there exists an orthogonal direct sum decomposition $\C^n=U_1\oplus 
\dots\oplus U_m$ such that for any $i\in[m]$, and any $B\in \sB$, we have 
$B(U_{i\boxminus 1})\leq U_i$, where $\boxminus$ indicates subtraction modulo $m$ 
in the range of $[m]$. 
\end{definition}

The following lemma relates isotropic $2$-decompositions with periodicity.
\begin{lemma}\label{lem:iso_period}
Given $\sB\in\IQC(n, \C)$, $\sB$ admits an isotropic $2$-decomposition if and only 
if the period of $\sB$ is $2k$ for some integer $k$.
\end{lemma}
\begin{proof}
For the if direction, let $\C^n=U_1\oplus U_2\oplus \dots \oplus U_{2k}$ be the 
orthogonal direct sum decomposition corresponding to the period of $\sB$. Let 
$V_1=\langle U_i : i=2j-1, j\in[k]\rangle$, and $V_2=\langle U_i : i=2j, 
j\in[k]\rangle$. Then $V_1\oplus V_2$ is an orthogonal direct sum decomposition, 
and for any $B\in \sB$, $B(V_1)\leq V_2$, and $B(V_2)\leq V_1$. By the orthogonal 
condition, $v_1^\dagger v_2=0$ for any $v_1\in V_1, v_2\in V_2$. We then have for 
any 
$i=1, 2$, any $v_i, v_i'\in V_i$, and any $B\in \sB$, we have $v_i^\dagger 
Bv_i'=0$. That 
is, $V_1$ and $V_2$ are isotropic spaces. 

For the only if direction, let $\C^n=V_1\oplus V_2$ be an isotropic 
$2$-decomposition. Let $P_1$ be the projection into $V_1$ along $V_2$, and $P_2$ 
the projection into $V_2$ along $V_1$. We have $P_1+P_2=I$, and $P_i^\dagger=P_i$. 
Since $V_1$ and $V_2$ are isotropic spaces, for any $B\in \sB$, and any $i=1, 2$, 
$P_iBP_i=\zerovec$. Using $P_1+P_2=I$, it follows that $P_2B=BP_1$, and 
$P_1B=BP_2$. We are then in the position to apply \cite[Lemma 
4.2]{fagnola2009irreducible}, to conclude that the period of $\sB$ is $2k$ for 
some integer $k$.
\end{proof}
%

Given Lemma~\ref{lem:iso_period}, it is enough to compute the period of $\sB$, and 
this can be done by resorting to the algorithm in \cite{guan2018decomposition}. 
For completeness, we give a brief sketch of the idea. 
 By Lemma 13 of \cite{guan2018decomposition},
 the period of irreducible quantum channel is equivalent to be the
number of eigenvalues with magnitude one of the quantum channel. Using the 
terminologies in the present article, we have the following lemma.
\begin{lemma}[{\cite[Lemma 13]{guan2018decomposition}}]  Given $\sB\in \IQC(n, 
\C)$, the 
period of $\sB$ is equal to
  the number of eigenvalues of $\tilde\sB$ with magnitude one.
\end{lemma}
Given this lemma, we can explicitly write out the form of $\tilde\sB$ as 
an $n^2\times n^2$ matrix, and compute its eigenvalues using e.g. \cite{Cai94} in 
the exact model (Section~\ref{subsec:algo_model}). 
Therefore, each eigenvalue $\alpha$ is represented by an irreducible polynomial 
$f(x)$ and a separating rectangle in the complex plane. To decide whether $\alpha$ 
has 
magnitude $1$ can be done efficiently by resorting to techniques from 
\cite{lovasz1986algorithmic}.
\end{proof}

Finally, we remark that the investigation in this subsection is not completely 
satisfactory. It would be more satisfying to consider isotropic spaces and 
isotropic 
decompositions for arbitrary quantum channels, not just the irreducible ones. We 
adopt the current strategy, partly because for irreducible channels, the 
periodicity is well-studied and well-connected with isotropic $2$-decomposition. 
We leave it a future work to study isotropic spaces 
and decompositions in the general setting. 
%
\subsection{Quantum gate subspace-fidelity and isotropic spaces}
We provide one quantum information theoretic interpretation for 
isotropic spaces, by relating it to quantum gate (state) fidelity~\cite[Section 
9]{nielsen2002quantum} and noiseless subspaces in quantum error 
correction~\cite{KL97,lidar2012review}. For 
the sake of readers who have little quantum information knowledge, we shall 
proceed by introducing all the necessary notions from quantum information, even 
though most of them are standard. 

In quantum information theory, the fidelity is a measure of the ``closeness'' of 
two quantum states, generalizing the fidelity of two 
distributions over finite events. It expresses the probability that one state will 
pass a test (quantum measurements) to identify as the other. Formally, the 
fidelity of two quantum states $\rho,\sigma\in \D(n, \C)$ is defined by
$$F(\rho,\sigma)=[\tr(\sqrt{\sqrt{\rho}\sigma \sqrt{\rho}})]^2.$$
It is worth noting that $0\leq F(\rho,\sigma)\leq 1$. Furthermore,
\begin{itemize}
	\item $F(\rho,\sigma)=0$ if and only if $\rho$ and $\sigma$ are orthogonal, i.e., $\tr(\rho\sigma)=0$;
	\item  $F(\rho,\sigma)=1$ if and only if $\rho=\sigma$.
\end{itemize}

Quantum state fidelity induces quantum gate fidelity. Unitary channels (i.e. 
channels of the form $\tilde{V}(A)=VAV^\dagger$ for some unitary matrix $V\in 
M(n,\C)$) are exactly the 
channels that do not introduce mixedness (i.e., decoherence) into states. 
Therefore,  
in experimental settings, they are considered to be the ideal type of channels  to 
be 
implemented~\cite[Section 8]{nielsen2002quantum}. However, no implementation of a 
channel is perfect, as there is no closed (isolated) system, so environment errors 
are unavoidable, which cause the channel actually implemented to be not
unitary. 
The gate fidelity is a tool for comparing how well the implemented quantum channel 
$\tilde{\sB}$ approximates the desired unitary channel $\tilde{V}$. Specifically, 
the gate fidelity on a pure state ($uu^\dagger$ for a normalized vector $u\in 
\C^n$) is a function defined as follows:
$$F_{\tilde{\sB},\tilde{V}}(u)=F(\tilde{\sB}(uu^\dagger),Vuu^\dagger V^\dagger)=u^\dagger V^\dagger\tilde{\sB}(uu^\dagger)Vu=u^\dagger[\tilde{V}^\dagger\circ\tilde{\sB}(uu^\dagger)]u,$$
where $\tilde{V}^\dagger(A)=V^\dagger AV$. In particular,  
$F_{\tilde{\sB},\tilde{V}}(u)=F_{\tilde{V}^
\dagger\circ \tilde{\sB},\tilde{I}}(u)$, where $\tilde{I}$ is the identity channel.
Then the gate fidelity on all states is defined as follows:
$$F(\tilde{\sB})=\min_{u\in \C^n} F_{\tilde{\sB},\tilde{I}}(u)=\min_{\rho\in D(n,\C)}F_{\tilde{\sB},\tilde{I}}(\rho)$$
The second equation in the above follows from the joint concavity property of the 
state fidelity $F$ (see \cite[Equation 9.121]{nielsen2002quantum}).

As we can see, quantum gate fidelity is a global property over $\C^n$. But in some 
cases, we only need a subspace of $\C^n$ as the state space of quantum information 
processing. This consideration motivates the following notions, which we 
call them quantum gate maximum and minimum subspace-fidelities, respectively.
For a subspace $U\subseteq\C^n$, 
$$F_{U}^{min}(\tilde{\sB})=\min_{u\in U}F_{\tilde{\sB},\tilde{I}}(u), \qquad 
\text{ and } 
F_{U}^{max}(\tilde{\sB})=\max_{u\in U}F_{\tilde{\sB},\tilde{I}}(u). $$
Note that
$F_{U}^{min}(\tilde{\sB})$ and $F_{U}^{max}(\tilde{\sB})$ quantify the worst-case 
and best-case behavior of the system by minimizing and maximizing over all 
possible initial states, respectively. Obviously, $0\leq 
F_{U}^{min}(\tilde{\sB})\leq F_{U}^{max}(\tilde{\sB})\leq1$ as $0\leq 
F(\rho,\sigma)\leq 1$. We also have 
$F(\tilde{\sB})=\min_{U\subseteq\C^n} 
F^{min}_U(\tilde{\sB})=\min_{U\subseteq\C^n} F^{max}_U(\tilde{\sB})$, where the 
second equation follows by examining one-dimensional subspaces.

One key notion in quantum error 
correction is that of noiseless subspaces, which have been intensively discussed 
in the setting where $\tilde{\sB}$ as a noise 
model~\cite{KL97,lidar2012review}. Intuitively, noiseless subspaces are shelters 
under quantum noise 
$\tilde{\sB}$, as they perfectly preserve quantum states under $\tilde{\sB}$. 
\begin{definition}
Let $\sB=\{B_1, \dots, B_m\}\in\QC(n, 
\C)$. An noiseless subspace of $\sB$ is a non-zero subspace $U\leq \C^n$, such that for any 
$u\in U$, and any $B_i$, we have $ B_iu=u$. 
\end{definition}


The following result formally shows that noiseless subspaces and isotropic 
subspaces are totally opposite from the viewpoint of the two quantum gate 
subspace-fidelities.

\begin{proposition}\label{prop:quantum}
Let $\sB=\{B_1, \dots, B_m\}\in\QC(n, 
\C)$.
\begin{itemize}
	\item $F_U^{max}(\tilde{\sB})=0$ if and only if $U$ is an isotropic space;
	\item $F_U^{max}(\tilde{\sB})=1$ if and only if there is a noiseless subspace in $U$;
	\item $F_U^{min}(\tilde{\sB})=1$ if and only if $U$ is a noiseless subspace;
	\item $F_U^{min}(\tilde{\sB})=0$ if and only if there is an isotropic space in $U$.
\end{itemize}
\end{proposition}
\begin{proof}
	These four claims are directly from the definitions of isotropic spaces, noiseless subspaces, quantum gate minimum subspace-fidelity and maximum subspace-fidelity.
\end{proof}

Knill devised an efficient algorithm to find all noiseless 
subspaces for a given $\tilde{\sB}$~\cite{knill2006protected}. So we have a quite 
good  
understanding on 
$F_U^{min}(\tilde{\sB})=1$. On the other hand, isotropic subspaces fully 
characterize  $F_U^{max}(\tilde{\sB})=0$. Therefore, isotropic 
spaces 
reveal the structure of the worst-case behavior of the channel.   

Let us further point out another potential application of isotropic spaces  in 
quantum control. 
A basic task of controlling quantum systems is to transfer all 
unknown quantum states into some targeting 
subspace~\cite{ticozzi2008quantum,cirillo2015decompositions}. So designing a 
control scheme as a quantum channel with a 
non-trivial isotropic space (the dimension greater than 1) can turn all quantum 
states residing in the isotropic space into the orthogonal complement of it.

\appendix


\section{Breadth-first search in the alternating matrix space 
setting}\label{app:breadth}

Indeed, suppose $\cA\leq 
\Lambda(n, \F)$ admits an isotropic $2$-decomposition as $\F^n=U_1\oplus U_2$. 
Note that $U_1$ and $U_2$ are not known to us. To 
follow the idea of breadth-first search, we would start from a vector $v\in\F^n$, 
and then find its neighbours, and then its neighbours' neighbours, etc.. 
Intuitively, for $v\in \F^n$, we can view the linear span of $Av$, $A\in \cA$ as 
those neighbours of $v$, denoted by $V_1\leq \F^n$. Then the linear span of 
$AV_1$, $A\in \cA$, may be considered as the neighbours of $V_1$. Continuing this 
way, if $v\in U_1$, we do see that $V_i$'s alternates between subspaces of $U_1$ 
and $U_2$. It follows that $V_i\cap V_{i+1}=\mathbf{0}$, from which we can compute 
$U_1$ and $U_2$ after this sequence stabilizes. However, if $v$ is neither in 
$U_1$ nor in $U_2$, it is not clear how to read any information about $U_1$ and 
$U_2$. In fact, it is possible that the linear span of $Av$ is the hyperplane 
orthogonal to $v$, so 
it is impossible to tell whether such $U_1$ and $U_2$ exist.

%

\section{The relation between alternating bilinear maps and alternating matrix 
spaces}\label{app:spaces_maps}

We first recall the relation between alternating bilinear maps and alternating 
matrix tuples. Let $\phi:U\times U\to V$ be an alternating bilinear map, that is, 
for any $u\in U$, $\phi(u, u)=\zerovec$. Fix bases of $U$ and $V$, so that $U\cong 
\F^n$ and $V\cong \F^m$. Then $\phi$ can be represented by an $m$-tuple of 
alternating matrices $(A_1, \dots, A_m)\in \Lambda(n, \F)^m$, such that $\phi(u, 
u')=(u^tA_1u', \dots, u^tA_mu')^t$. Conversely, given an $m$-tuple of alternating 
matrices, one can define an alternating bilinear map as such. Two alternating 
bilinear maps $\phi, \psi:U\times U\to V$ are isometric, if there exist $A\in 
\GL(U)$, $B\in \GL(V)$, such that $\phi=B\circ \psi\circ A$. (Some authors prefer 
to call this isometric as pseudo-isometric \cite{BW12}.)

Let $\cA\leq \Lambda(n, \F)$. Let $\bA=(A_1, \dots, A_m)\in \Lambda(n, \F)^m$ be 
an ordered basis of $\cA$. Then $\bA$ defines an alternating bilinear map 
$\phi_\bA:\F^n\times \F^n\to \F^m$ as above. While difference choices of ordered 
bases give different alternating bilinear maps, it is easy to see that ordered 
bases from isometric alternating matrix spaces give isometric alternating bilinear 
maps.

\paragraph{Acknowledgement.} The authors would like to thank Nengkun Yu, James B. 
Wilson, and G\'abor
Ivanyos for discussions related to this paper. They would also like to thank 
George 
Glauberman and L\'aszl\'o Pyber for answering their questions on groups, including 
pointing out the reference \cite{Ols78} and clarifying the field characteristic 
issue in \cite{BGH87}.

Y. Q. was partly supported by the Australian Research
Council DECRA DE150100720. Y. Q. would like to thank Nengkun Yu, James B. Wilson,
and G\'abor
Ivanyos for discussions related to this paper. Y. Q. would also like to thank 
George 
Glauberman and L\'aszl\'o Pyber for answering his questions on groups, including 
pointing out the reference \cite{Ols78} and clarifying the field characteristic 
issue in \cite{BGH87}. J. 
G. was partly supported by the
National Key R\&D Program of China (Grant No: 2018YFA0306701), the National
Natural Science Foundation of China (Grant No: 61832015). S. C. was partly 
supported by National Natural Science Foundation of China (Grant No:61702489).

\bibliographystyle{alpha}
\bibliography{references}

\end{document}